\documentclass[10pt]{article}
\usepackage{epsfig,amsmath,graphicx,amssymb,overpic}

\usepackage{graphicx}
\usepackage{dcolumn}
\usepackage{bm}
\usepackage{graphicx}
\usepackage{subfigure}
\usepackage{epsfig,amsmath,graphicx,amssymb,overpic,cite}

\usepackage{graphicx}
\usepackage{dcolumn}
\usepackage{bm}
\usepackage{graphicx}
\usepackage{subfigure}
\usepackage{amsthm}

\newtheorem{proposition}{Proposition}

\newtheorem{theorem}{Theorem}

\def\be{\begin{equation}}
\def\ee{\end{equation}}
\def\bee{\begin{eqnarray}}
\def\ene{\end{eqnarray}}
\def\bes{\begin{subequations}}
\def\ees{\end{subequations}}

\def\v{\vspace{0.1in}}
\def\no{{\nonumber}}

\begin{document}

\baselineskip=13pt
\renewcommand {\thefootnote}{\dag}
\renewcommand {\thefootnote}{\ddag}
\renewcommand {\thefootnote}{ }

\pagestyle{plain}

\begin{center}
\baselineskip=16pt \leftline{} \vspace{-.3in} {\Large \bf Inverse scattering transforms and $N$-double-pole solutions for the derivative NLS equation with zero/non-zero boundary conditions} \\[0.2in]
\end{center}

\begin{center}
Guoqiang Zhang$^{1,2}$ and Zhenya Yan$^{1,2,*}$\footnote{$^{*}${\it Email address}: zyyan@mmrc.iss.ac.cn}  \\[0.03in]
{\it \small $^{1}$Key Laboratory of Mathematics Mechanization, Academy of Mathematics and Systems Science, \\ Chinese Academy of Sciences, Beijing 100190, China \\
 $^{2}$School of Mathematical Sciences, University of Chinese Academy of Sciences, Beijing 100049, China} \\
\end{center}

\vspace{0.1in}

{\baselineskip=13pt

\begin{tabular}{p{16cm}}
 \hline \\
\end{tabular}

\vspace{-0.18in}

\begin{abstract} \small \baselineskip=12pt

We systematically report a rigorous theory of the inverse scattering transforms (ISTs) for the derivative nonlinear Schr\"odinger (DNLS) equation with both zero boundary condition (ZBC)/non-zero boundary conditions (NZBCs) at infinity and double poles of analytical scattering coefficients. The scattering theories for both ZBC and NZBCs are addressed. The direct problem establishes the analyticity, symmetries and asymptotic behavior of the Jost solutions and scattering matrix, and properties of discrete spectra. The inverse problems are formulated and solved with the aid of the matrix Riemann-Hilbert problems, and the reconstruction formulae, trace formulae and theta conditions are also posed. In particular, the IST with NZBCs at infinity is proposed by a suitable uniformization variable, which allows the scattering problem to be solved on a standard complex plane instead of a two-sheeted Riemann surface. The reflectionless potentials with double poles for the ZBC and NZBCs are both carried out explicitly by means of determinants. Some representative semi-rational bright-bright soliton, dark-bright soliton, and breather-breather solutions  are
examined in detail. These results will be useful to further explore and apply the related nonlinear wave phenomena.

\vspace{0.1in} \noindent {\it Keywords:}  inverse scattering; Riemann-Hilbert problem; derivative nonlinear Schr\"odinger equation; zero/non-zero boundary conditions; double-pole solitons and breathers

\end{abstract}

\vspace{-0.05in}
\begin{tabular}{p{16cm}}
  \hline \\
\end{tabular}

\vspace{-0.15in}

\section{Introduction}

As a fundamental and important nonlinear physical model, the derivative nonlinear Schr\"odinger (DNLS) equation
\begin{gather}\label{DNLS}
\begin{array}{l}
iq_t+q_{xx}+i\sigma(\left|q\right|^2q)_x=0,
\end{array}
\end{gather}
has several physical applications, such as the propagation of circular polarized nonlinear Alfv{\' e}n waves in plasmas \cite{Rogister1971,Mjolhus1976, Mio1976, Mjolhus1989, Mjolhus1997}, weak nonlinear electromagnetic waves in ferromagnetic \cite{Nakata1991}, antiferromagnetic \cite{Daniel2002} or dielectric \cite{Nakata1993} systems under external magnetic fields. The parameter $\sigma$ stands for the relative magnitude of the derivative nonlinearity term, without loss of generality, one can take $\sigma=-1$ (since the case  $\sigma=1$ can be transformed into $\sigma=-1$ by means of $x\to -x$). Eq.~(\ref{DNLS}) can be transformed into the modified NLS equation
 \be\label{dnlsg}
iu_{\tau}+u_{\xi\xi}+g|u|^2u+i\gamma(|u|^2u)_{\xi}=0,
\ee
by means of the gauge transformation~\cite{gauge} $u(\tau, \xi)=q(x, t)e^{i(kx+k^2t)}$ with $x=\frac{\gamma}{\sigma}\xi-\frac{2g}{\sigma}\tau$,\, $t=\frac{\gamma^2}{\sigma^2}\tau$, and $k=\frac{\sigma g}{\gamma^2}$, where the Kerr nonlinear coefficient $\lambda$ and derivative nonlinear coefficient $\gamma$ both depend on nonlinear refractive index $n_2$. The modified NLS equation (\ref{dnlsg}) describes transmission of
femtosecond pulses in optical fibers~\cite{d1,d2,d3}. The solutions of nonlinear wave equations with ZBC and NZBC are always physically interesting subjects.

The  inverse scattering transform (IST) due to Gardner, Greene, Kruskal and Miura~\cite{Gardner1967} provides a powerful approach to discover solutions and properties of some integrable nonlinear wave systems. The long-time leading-order asymptotics of Eqs.~(\ref{DNLS}) and (\ref{dnlsg}) were studied~\cite{RH1,RH3} by the Deift-Zhou method~\cite{RH2}. The IST for the DNLS equation with ZBC has been studied to find its one-soliton solution~\cite{Kaup1978} and $N$-soliton solutions~\cite{Zhou2007}. The ISTs for the DNLS equation with NZBCs at infinity were also developed~\cite{Kawata1978, Chen2004, Chen2006, Lashkin2007, Zhou2012}. However, to the best of our knowledge, these IST works on the DNLS equation only focus on the case that all discrete spectra are {\it simple}. Moreover, no trace formulae or theta conditions were researched, the inverse problems were not formulated by means of Riemann-Hilbert problems. A natural problem is whether explicit {\it multi-pole} solutions can be found for the DNSL equation with ZBC/NZBCs by the approximate IST based on the matrix Riemann-Hilbert problems~\cite{Prinari2006,Demontis2013, Demontis2014,Biondini2014,Prinari2015a, Pichler2017, zhang18}.

In the present paper, we present the ISTs for the DNLS equation with ZBC/NZBCs based on the matrix Riemann-Hilbert problems (RHPs), and present their novel {\it double-pole} solutions by solving the corresponding RHPs. It should pointed out the DNLS equation is associated with the modified Zakharov-Shabat eigenvalue problem, not the usual Zakharov-Shabat eigenvalue problem related to the NLS-type equations~\cite{Shabat1972,Ablowitz1973, Prinari2006,Demontis2013, Demontis2014,Biondini2014,Prinari2015a, Pichler2017, zhang18} such that the discrete spectrum and solving Riemann-Hilbert problems are more complicated.


DNLS equation (\ref{DNLS}) is completely integrable and associated with the following modified Zakharov-Shabat eigenvalue problem \cite{Kaup1978}:
 \begin{align}\label{lax-x}
\varPhi_x&=X\varPhi, \quad X(x,t; k)=ik^2\sigma_3+kQ, \\[0.01in] \label{lax-t}
\varPhi_t&=T\varPhi, \quad \,\,T(x,t; k)=-\left(2k^2+Q^2\right)X-ikQ_x\sigma_3,
\end{align}
where $Q=Q(x,t)$ is written as
\begin{gather}
Q=\begin{bmatrix}
0& q(x,t)\\
\sigma\,q^*(x,t)&0
\end{bmatrix},
\end{gather}
and $\sigma_3$ is one of the Pauli's matrices, which are
\begin{align} \no
\sigma_1=\begin{bmatrix}
0&1\\1&0
\end{bmatrix},\quad
\sigma_2=\begin{bmatrix}
0&-i\\i&0
\end{bmatrix},\quad
\sigma_3=\begin{bmatrix}
1&0\\
0&-1
\end{bmatrix}.
\end{align}

The rest of this paper is organized as follows. In Sec. II, the IST for DNLS with ZBC at infinity is introduced, and solved for the double poles of analytical scattering coefficients by means of the matrix Riemann-Hilber problem. As a consequence, we present a formula for the explicit double-pole $N$-soliton solutions. In Sec. III, we give a detailed theory of the IST for the DNLS equation with NZBCs at infinity, which is more complicated than the case of ZBC since more symmetries and a two-sheeted Riemann surface are required. As a result, we present an explicit formula for the double-pole $N$-solitons for the case of NZBCs. Particularly, we discuss the special double-pole solitons. Finally, the conclusions and discussions are carried out in Sec. IV.

\section{IST with ZBC and double poles}

In this section, we will seek for a solution $q(x, t)$ for DNLS equation (\ref{DNLS}) with $\sigma=-1$ and ZBC
\begin{gather}\label{ZBC}
q(x, t)\to 0,\quad \mathrm{as}\quad x\to \pm\infty.
\end{gather}
The IST for DNLS equation (\ref{DNLS}) with ZBC (\ref{ZBC}) was first presented by Kaup and Newell \cite{Kaup1978}, where the simple poles for reflection coefficients are required. In what follows, we will further present the IST and solitons for Eq.~(\ref{DNLS}) with ZBC and {\it double poles}.

\subsection{Direct scattering with ZBC}

\subsubsection{Jost solutions, analyticity, and continuity}

Considering the asymptotic scattering problem ($x\to\pm\infty$) of the Lax pair (\ref{lax-x}, \ref{lax-t})
\begin{gather} \label{laxs}
\begin{aligned}
\varPhi_x&=X_0 \varPhi, \quad X_0=ik^2\sigma_3, \\[0.01in]
\varPhi_t&=T_0\varPhi, \quad \,\,T_0=-2ik^4\sigma_3=-2k^2X_0,
\end{aligned}
\end{gather}
the fundamental matrix solution $\varPhi^{bg}(x, t; k)$ of system (\ref{laxs}) can be derived as
\begin{gather}
\varPhi^{bg}(x, t; k)=\mathrm{e}^{i\theta(x, t; k)\sigma_3}, \quad \theta(x, t; k)=k^2\left(x-2k^2 t\right).
\end{gather}

 Let $\Sigma:=\mathbb{R}\cup i\mathbb{R}$. We will seek for the Jost solutions $\varPhi_{\pm}(x, t; k)$ such that
\begin{gather}\label{ZBC-Jxjianjin}
\varPhi_{\pm}(x, t; k)=\mathrm{e}^{i\theta(x, t; k)\sigma_3}+o\left(1\right), \quad k\in\Sigma, \quad \mathrm{as} \quad  x\to\pm\infty.
\end{gather}
Consider the modified Jost solutions $\mu_{\pm}(x, t; k)$ in the form
\begin{gather}
\mu_{\pm}(x, t; k)=\varPhi_{\pm}(x, t; k)\,\mathrm{e}^{-i\theta(x, t; k)\sigma_3},
\end{gather}
which leads to $\mu_{\pm}(x, t; k)\to I$ as $x\to \pm \infty$, then we know that $\mu_{\pm}(x,t;k)$ satisfy the following Jost integrable equations
\begin{gather}
\mu_{\pm}(x, t; k)=I+k\int_{\pm\infty}^x \mathrm{e}^{ik^2(x-y)\sigma_3}Q(y, t)\,\mu_{\pm}(y, t; k)\,\mathrm{e}^{-ik^2(x-y)\sigma_3}\,\mathrm{d}y.
\end{gather}
which differs from ones for the NLS equation with ZBC related to the Zakharov-Shabat eigenvalue problem.

Let
\begin{gather} \no
D^+:=\left\{k\in \mathbb{C}\,|\,\mathrm{Re}(k)\mathrm{Im}(k)>0\right\}, \qquad D^-:=\left\{k\in \mathbb{C}\,|\,\mathrm{Re}(k)\mathrm{Im}(k)<0\right\},
\end{gather}
which are shown in Fig. \ref{k-ZBC}. We have the following proposition.
\begin{figure}[!t]
\centering
\includegraphics[scale=0.35]{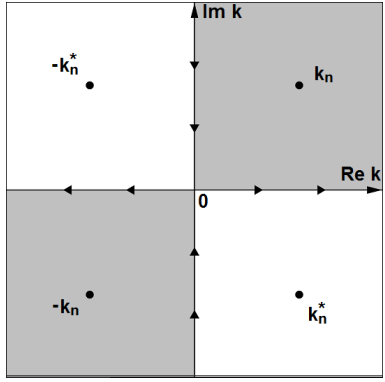}
\caption{Complex $k$-plane. Region $D^+$ (grey region), region $D^-$ (white region), discrete spectra, and orientation of the contours for the Riemann-Hilbert problem.}
\label{k-ZBC}
\end{figure}

\begin{proposition}
Suppose that $q(x,t)\in L^1\left(\mathbb{R}\right)$ and $\varPhi_{\pm j}(x, t; k)$ is the $j$th column of $\varPhi_{\pm}(x, t; k)$, then $\varPhi_\pm(x, t; k)$ have the following  properties:
\begin{itemize}
\item Eq. (\ref{lax-x}) has the unique Jost solutions $\varPhi_{\pm}(x, t; k)$ satisfying (\ref{ZBC-Jxjianjin}) on $\Sigma$;

\item $\varPhi_{+1,-2}(x, t; k)$ can be analytically extended  to $D^{+}$ and continuously extended to $D^{+}\cup\Sigma$;

\item $\varPhi_{-1, +2}(x, t; k)$  can be analytically extended  to $D^{-}$ and continuously extended to $D^-\cup\Sigma$.
\end{itemize}
\end{proposition}
The analyticity and continuity for the modified Jost solutions $\mu_{\pm}(x, t; k)$ follow trivially from those of $\varPhi_{\pm}(x, t; k)$.

\begin{proposition}[Time evolution of the Jost solutions]
The Jost solutions $\varPhi_{\pm}(x, t; k)$ are the simultaneous solutions of both parts of the Lax pair (\ref{lax-x},  \ref{lax-t}).
\end{proposition}
\begin{proof}
We just need to illustrate that $\varPhi_{\pm}(x, t; k)$ solve the $t$-part (\ref{lax-t}). Liouville's formula leads to
\begin{align*}
\mathrm{det}\,\varPhi_{\pm}(x, t; k)=\lim_{x\to\pm\infty}\mathrm{det}\,\varPhi_{\pm}(x, t; k)=\lim_{x\to\pm\infty}\mathrm{det}\,\mu_{\pm}(x, t; k)=1.
\end{align*}
That is to say, $\varPhi_{\pm}(x, t; k)$ are the fundamental matrix solutions on $\Sigma$. By the compatibility condition $X_t-T_x+[X, T]=0$, one can find that $\varPhi_{\pm t}(x, t; k)-T\varPhi_{\pm}(x, t; k)$  also solve the $x$-part (\ref{lax-x}). Thus, there exist  two matrices $G_{\pm}(t; k)$ such that
\begin{gather*}
\varPhi_{\pm t}(x, t; z)-T\varPhi_{\pm}(x, t; z)=\varPhi_{\pm}(x, t; z)\,G_{\pm}(t; z),\quad \mathrm{as}\quad k\in\Sigma.
\end{gather*}
Multiplying both sides by $\mathrm{e}^{-i\theta(x, t; z)\sigma_3}$ and letting $x\to\pm\infty$, one  obtains  $G_{\pm}(t; k)=0$. The proof follows.
\end{proof}

\subsubsection{Scattering matrix and reflection coefficients}

Since $\varPhi_{\pm}(x, t; k)$ are two fundamental matrix solutions,  there exists a constant matrix $S(k)$ (not depend on $x$ and $t$) such that
\begin{align}\label{ZBC-ssjz}
\varPhi_+(x, t; k)=\varPhi_-(x, t; k)\,S(k), \quad  k\in\Sigma,
\end{align}
where $S(k)=\left(s_{ij}(k)\right)_{2\times 2}$ is referred to the scattering matrix and its entries $s_{ij}(k)$ as the scattering coefficients.
It follows from Eq.~(\ref{ZBC-ssjz}) that $s_{ij}(z)$'s have the Wronskian representations:
\begin{gather}\label{ZBC-lsj}
\begin{aligned}
s_{11}(k)={\rm Wr}(\varPhi_{+1}(x, t; k), \varPhi_{-2}(x, t; k)),\quad s_{22}(z)={\rm Wr}(\varPhi_{-1}(x, t; z), \varPhi_{+2}(x, t; z)),\\[0.05in]
s_{12}(z)={\rm Wr}(\varPhi_{+2}(x, t; z), \varPhi_{-2}(x, t; z)),\quad s_{21}(z)={\rm Wr}(\varPhi_{-1}(x, t; z), \varPhi_{+1}(x, t; z)).
\end{aligned}
\end{gather}

\begin{proposition}
Suppose that $q(x,t)\in L^1\left(\mathbb{R}\right)$, then $s_{11}(k)$ can be analytically extended  to $D^+$ and continuously extended to $D^+\cup\Sigma$ while $s_{22}(k)$ can be analytically extend  to $D^-$ and continuously extended to $D^-\cup\Sigma^0$. Moreover, both $s_{12}(k)$ and $s_{21}(k)$ are continuous in $\Sigma$.
\end{proposition}

Note that one cannot exclude the possible presence of zeros for $s_{11}(k)$ and $s_{22}(k)$ along $\Sigma$. To solve the Riemann-Hilbert problem in the inverse process, we restrict our consideration to the potential without spectral singularities \cite{Zhou1989}. As usual, the reflection coefficients $\rho(k)$ and $\tilde\rho(k)$ are  defined as
\begin{align}
\rho(k)=\frac{s_{21}(k)}{s_{11}(k)},\quad
\tilde\rho(k)=\frac{s_{12}(k)}{s_{22}(k)}, \quad  k\in\Sigma.
\end{align}

\subsubsection{Symmetries}

\begin{proposition}[Symmetry reduction~\cite{Zhou2007}]\, $X(x, t; k)$,\, $T(x, t; k)$, Jost solutions, scattering matrix, and reflection coefficients keep two symmetry reductions
\begin{itemize}
\item The first symmetry reduction
\begin{gather} \no
\begin{gathered}
X(x, t; k)=\sigma_2\,X(x, t; k^*)^*\sigma_2, \quad T(x, t; k)=\sigma_2\,T(x, t; k^*)^*\sigma_2, \v\\
\varPhi_{\pm}(x, t; k)=\sigma_2\,\varPhi_{\pm}(x, t; k^*)^*\,\sigma_2,\quad S(k)=\sigma_2\,S(k^*)^*\,\sigma_2, \quad \rho(k)=-\tilde\rho(k^*)^*.
\end{gathered}
\end{gather}
\item The second symmetry reduction
\begin{gather} \no
\begin{gathered}
X(x, t; k)=\sigma_1X(x, t; -k^*)^*\sigma_1,\quad T(x, t; k)=\sigma_1T(x, t; -k^*)^*\sigma_1, \v\\
\varPhi_{\pm}(x, t; k)=\sigma_1\,\varPhi_{\pm}(x, t; -k^*)^*\,\sigma_1, \quad S(k)=\sigma_1\,S(-k^*)^*\,\sigma_1, \quad \rho(k)=\tilde\rho(-k^*)^*.
\end{gathered}
\end{gather}
\end{itemize}
\end{proposition}

\subsubsection{Discrete spectrum with double poles}

The discrete spectrum of the scattering problem is the set of all values $k\in \mathbb{C}\backslash\Sigma$ such that the scattering problem admits eigenfunctions in $L^2(\mathbb{R})$. As was shown in~\cite{Biondini2014}, these exist exactly the values of $k$ in $D^+$ such that $s_{11}(k)=0$ and those values in $D^-$ such that $s_{22}(k)=0$. Differing from the previous results with single poles~\cite{Kaup1978, Zhou2007}, we here suppose that $s_{11}(k)$ has $N$ double zeros in $\left\{k\in\mathbb{C}:\mathrm{Re} \,k>0, \mathrm{Im}\, k>0\right\}$ denoted by $k_n$, $n=1, 2, \cdots, N$, that is, $s_{11}(k_n)=s'_{11}(k_n)=0$ and $s''_{11}(k_n)\not=0$. It follows from the symmetries of the scattering matrix that
\begin{gather}
\begin{gathered}
s_{11}(k_n)=s_{11}(-k_n)=s_{22}(k_n^*)=s_{22}(-k_n^*)=0,\\[0.05in]
s'_{11}(k_n)=s'_{11}(-k_n)=s'_{22}(k_n^*)=s'_{22}(-k_n^*)=0.
\end{gathered}
\end{gather}

The discrete spectrum is the set
\begin{align}\label{lisanpu}
 K=\left\{k_n,\, k_n^*,\, -k_n^*, \,-k_n\right\}_{n=1}^{N},
\end{align}
whose distributions are shown in Fig. \ref{k-ZBC}. Given $k_0\in K\cap D^+$, it follows from the Wronskian representations (\ref{ZBC-lsj}) and $s_{11}(k_0)=0$ that $\varPhi_{+1}(x, t; k_0)$ and $\varPhi_{-2}(x, t; k_0)$ are linearly dependent.  Given $k_0\in Z\cap D^-$, it follows from the Wronskian representations (\ref{ZBC-lsj}) and $s_{22}(k_0)=0$ that $\varPhi_{+2}(x, t; k_0)$ and $\varPhi_{-1}(x, t; k_0)$ are linearly dependent.
For convenience, we define $b[k_0]$ as the proportionality coefficient:
\begin{gather}
b[k_0]=\left\{
\begin{aligned}
\frac{\varPhi_{+1}(x, t; k_0)}{\varPhi_{-2}(x, t; k_0)}, \quad k_0\in K\cap D^+,\\[0.05in]
\frac{\varPhi_{+2}(x, t; k_0)}{\varPhi_{-1}(x, t; k_0)}, \quad k_0\in K\cap D^-.
\end{aligned}\right.
\end{gather}
Given $k_0\in K\cap D^+$, it follows from the Wronskian representations (\ref{ZBC-lsj}) and $s'_{11}(k_0)=0$ that $\varPhi_{+1}'(x, t; k_0)-b[k_0]\,\varPhi_{-2}'(x, t; k_0)$ and $\varPhi_{-2}(x, t; k_0)$ are linearly dependent. In the same manner, as $k_0\in K\cap D^-$, $\varPhi_{+2}'(x, t; k_0)-b[k_0]\,\varPhi_{-1}'(x, t; k_0)$ and $\varPhi_{-1}(x, t; k_0)$ are linearly dependent.  For convenience, we define $d[k_0]$ as the proportionality coefficient
\begin{gather}
d[k_0]=\left\{
\begin{aligned}
\frac{\varPhi_{+1}'(x, t; k_0)-b[k_0]\,\varPhi_{-2}'(x, t; k_0)}{\varPhi_{-2}(x, t; k_0)},\quad k_0\in K\cap D^+,\\[0.05in]
\frac{\varPhi_{+2}'(x, t; k_0)-b[k_0]\,\varPhi_{-1}'(x, t; k_0)}{\varPhi_{-1}(x, t; k_0)},\quad k_0\in K\cap D^-.
\end{aligned}\right.
\end{gather}

Moreover, let
\begin{gather}
A[k_0]=\left\{
\begin{aligned}
\frac{2\,b[k_0]}{s_{11}''(k_0)},\quad k_0\in K\cap D^+,\\[0.05in]
\frac{2\,b[k_0]}{s_{22}''(k_0)},\quad k_0\in K\cap D^-,
\end{aligned}\right. \qquad
B[k_0]=\left\{
\begin{aligned}
\frac{d[k_0]}{b[k_0]}-\frac{s_{11}'''(k_0)}{3\,s_{11}''(k_0)},\quad k_0\in K\cap D^+,\\[0.05in]
\frac{d[k_0]}{b[k_0]}-\frac{s_{22}'''(k_0)}{3\,s_{22}''(k_0)},\quad k_0\in K\cap D^-,
\end{aligned}\right.
\end{gather}
then one has the following compact form
\begin{align}
\begin{aligned}
\mathop\mathrm{P_{-2}}\limits_{k=k_0}\left[\frac{\varPhi_{+1}(x, t; k)}{s_{11}(k)}\right]&=A[k_0]\,\varPhi_{-2}(x, t; k_0), \quad k_0\in K\cap D^+,\\[0.05in]
\mathop\mathrm{P_{-2}}\limits_{k=k_0}\left[\frac{\varPhi_{+2}(x, t; k)}{s_{22}(k)}\right]&=A[k_0]\,\varPhi_{-1}(x, t; k_0), \quad k_0\in K\cap D^-,\\[0.05in]
\mathop\mathrm{Res}\limits_{k=k_0}\left[\frac{\varPhi_{+1}(x, t; k)}{s_{11}(k)}\right]&=A[k_0]\left[\varPhi_{-2}'(x, t; k_0)+B[k_0]\, \varPhi_{-2}(x, t; k_0)\right], \quad k_0\in K\cap D^+,\\[0.05in]
\mathop\mathrm{Res}\limits_{k=k_0}\left[\frac{\varPhi_{+2}(x, t; k)}{s_{22}(k)}\right]&=A[k_0]\left[\varPhi_{-1}'(x, t; k_0)+B[k_0]\, \varPhi_{-1}(x, t; k_0)\right], \quad k_0\in K\cap D^-,
\end{aligned}
\end{align}
where $\mathop\mathrm{P_{-2}}\limits_{k=k_0}\left[\bm\cdot\right]$ denotes the coefficient of $O\left(\frac{1}{\left(k-k_0\right)^2}\right)$ term in the Laurent expansion of $\bm\cdot$ at $k=k_0$.

\begin{proposition}
Given $k_0\in K$, two symmetry relations for $A[k_0]$ and $B[k_0]$ are given as follows.
\begin{itemize}
\item The first symmetry relation
\begin{gather} \no
A[k_0]=-A[k_0^*]^*, \quad B[k_0]=B[k_0^*]^*.
\end{gather}
\item The second symmetry relation
\begin{gather}\no
A[k_0]=A[-k_0^*]^*, \quad B[k_0]=-B[-k_0^*]^*.
\end{gather}
\end{itemize}
\end{proposition}
By the two symmetry relations, one obtains the following constraints of discrete spectrum.
\begin{gather}
\begin{gathered}
A[k_n]=-A[k_n^*]^*=A[-k_n^*]^*=-A[-k_n],\\[0.05in]
B[k_n]=B[k_n^*]^*=-B[-k_n^*]^*=-B[-k_n].
\end{gathered}
\end{gather}

\subsubsection{Asymptotic behaviors}

To propose and solve the Riemann-Hilbert problem in the inverse problem, one has to determine the asymptotic behaviors of the modified Jost solutions and scattering matrix as $k\to\infty$. The standard Wentzel-Kramers-Brillouin (WKB) expansions are used to derive the asymptotic behavior of the modified Jost solutions. We consider the following ansatz for the expansions of the modified Jost solutions $\mu_{\pm}(x, t; k)$ as $k\to\infty$:
\begin{gather}
\mu_{\pm}(x, t; k)=\sum_{j=0}^n\frac{\mu_{\pm}^{[j]}(x, t)}{k^j}+O\left(\frac{1}{k^{n+1}}\right), \quad \mathrm{as}\quad k\to\infty,
\end{gather}
and substituting $\varPhi_{\pm}(x, t; k)=\mu_{\pm}(x, t; k)\,\mathrm{e}^{i\theta(x, t; k)\sigma_3}$ with these expansions into Eq. (\ref{lax-x}). By matching the $O\left(k^2\right)$ term, one obtains the off-diagonal parts $\left(\mu_{\pm}^{[0]}(x, t)\right)^{\mathrm{off}}=0$. It follows by matching the $O\left(k\right)$ term that $\left(\mu_{\pm}^{[1]}(x, t)\right)^{\mathrm{off}}=\frac{i}{2}\,\sigma_3Q(x, t)\,\mu_{\pm}^{[0]}(x, t)$. By matching the $O\left(1\right)$ term, one yields that $\mu_{\pm}^{[0]}(x, t)=C^{\mathrm{diag}}\,\mathrm{e}^{iv_{\pm}(x,t)\sigma_3}$, where
\begin{gather}
v_{\pm}(x, t)=\frac{1}{2}\int_{\pm\infty}^x\left| q(y, t)\right|^2\,\mathrm{d}y.
\end{gather}
Combining with the asymptotic behaviors of the modified Jost solutions $\mu_{\pm}(x, t; k)$ as $x\to\pm\infty$, one deduces the asymptotic behavior as $k\to\infty$.

\begin{proposition}
The asymptotic behaviors of the modified Jost solutions are
\begin{gather}
\mu_{\pm}(x, t; k)=\mathrm{e}^{iv_{\pm}(x, t)\sigma_3}+O\left(\frac{1}{k}\right), \quad \mathrm{as}\quad k\to\infty.
\end{gather}
\end{proposition}

From the definition or Wronskian presentations of scattering matrix and the asymptotic behaviors of the modified Jost solutions, the asymptotic behaviors of the scattering matrix can be yielded below.
\begin{proposition}
The asymptotic behavior of the scattering matrix is
\begin{gather}
S(k)=\mathrm{e}^{-iv_0\sigma_3}+O\left(\frac{1}{k}\right), \quad \mathrm{as}\quad k\to\infty,
\end{gather}
where $v_0$ reads as
\begin{gather}
v_0=\frac{1}{2}\int_{-\infty}^{+\infty}\left| q(y, t)\right|^2\,\mathrm{d}y.
\end{gather}
\end{proposition}
Note that $v_0$ does not depend on variable $t$. In fact, substituting $\varPhi_{\pm}(x, t; k)=\mu_{\pm}(x, t; k)\,\mathrm{e}^{i\theta(x, t; k)\sigma_3}$ with these expansions into Eq. (\ref{lax-t}) and matching the $O\left(k^4\right)$,  $O\left(k^3\right)$, $O\left(k^2\right)$, $O\left(k\right)$ and $O\left(1\right)$ in order, then $v_{0t}=0$ follows as $x\to\mp\infty$.

\subsection{Inverse problem with ZBC and double poles}

\subsubsection{Matrix Riemann-Hilbert problem}

According to the relation (\ref{ZBC-ssjz}) of two fundamental matrix solutions $\varPhi_{\pm}(x, t; k)$, we have the following Riemann-Hilbert problem:

\begin{proposition}
Let the sectionally meromorphic matrices
\begin{gather}
M(x, t; k)=\left\{
\begin{aligned}
\left(\frac{\mu_{+1}(x, t; k)}{s_{11}(k)},\, \mu_{-2}(x, t; k)\right),\quad k\in D^+, \\[0.05in]
\left(\mu_{-1}(x, t; k),\, \frac{\mu_{+2}(x, t; k)}{s_{22}(k)}\right), \quad k\in D^-,
\end{aligned}\right.
\end{gather}
and
\begin{gather}
M^{\pm}(x, t; k)=\lim_{\begin{matrix}
k'\to k\\k'\in D^{\pm}
\end{matrix}}M(x, t; k'), \quad k\in\Sigma.
\end{gather}
Then the multiplicative matrix Riemann-Hilbert problem is given as follows:
\begin{itemize}
\item Analyticity: $M(x, t; k)$ is analytic in $\left(D^+\cup D^-\right)\backslash K$ and has double poles in $K$.
\item Jump condition:
\begin{gather}
M^-(x, t; k)=M^+(x, t; k)\left(I-J(x, t; k)\right), \quad k\in\Sigma,
\end{gather}
where $J(x, t; k)$ is defined by
\begin{gather*}
J(x, t; k)=\mathrm{e}^{i\theta(x, t; k)\sigma_3}
\begin{bmatrix}
0&-\tilde\rho(k)\\[0.05in]
\rho(k)&\rho(k)\,\tilde\rho(k)
\end{bmatrix}\mathrm{e}^{-i\theta(x, t; k)\sigma_3}.
\end{gather*}
\item Asymptotic behavior:
\begin{gather}
M(x, t; k)=\mathrm{e}^{iv_-(x, t)\sigma_3}+O\left(\frac{1}{k}\right), \quad k\to\infty.
\end{gather}
\end{itemize}
\end{proposition}
The analyticity follows trivially from the analyticity of the modified Jost solutions and scattering data. The jump condition can be derived as one rearranges the terms in Eq. (\ref{ZBC-ssjz}). The asymptotic behaviors of the modified Jost solutions and scattering matrix can lead to that of $M(x, t; k)$. To solve the Riemann-Hilbert problem conveniently, we define
\begin{gather}
\eta_n=\left\{
\begin{aligned}
&k_n, & n&=1, 2, \cdots, N\\[0.05in]
-{}&k_{n-N}, &n&=N+1, N+2, \cdots, 2N
\end{aligned}\right.
\end{gather}
By subtract out the asymptotic values as $k\to\infty$ and the singularity contributions, one can  regularize the Riemann-Hilbert problem as a standard form. Then combining with Cauchy projectors and Plemelj's formulae, one can establish the solutions of the Riemann-Hilbert problem.

\begin{proposition}
The solution of the Riemann-Hilbert problem is given by
\begin{gather}
\begin{aligned}
&M(x, t; k)=\mathrm{e}^{iv_-(x, t)\sigma_3}+\frac{1}{2\pi i}\int_\Sigma\frac{(M^+J)(x, t; \zeta)}{\zeta-k}\,\mathrm{d}\zeta  \\
&+\sum_{n=1}^{2N}\left(C_n(k)\left[\mu_{-2}'(\eta_n)+\left(D_n+\frac{1}{k-\eta_n}\right)\mu_{-2}(\eta_n)\right],
\,\widehat C_n(k)\left[\mu_{-1}'(\eta_n^*)+\left(\widehat D_n+\frac{1}{k-\eta_n^*}\right)\mu_{-1}(\eta_n^*)\right]\right),
\end{aligned}
\end{gather}
where $k\in \mathbb{C}\backslash \Sigma$,
\begin{gather} \no
\begin{gathered}
C_n(k)=\frac{A[\eta_n]}{k-\eta_n}\,\mathrm{e}^{-2i\theta(x, t; \eta_n)}, \quad D_n=B[\eta_n]-2\,i\,\theta'(x, t; \eta_n), \\[0.05in]
\widehat C_n(k)=\frac{A[\eta_n^*]}{k-\eta_n^*}\,\mathrm{e}^{2i\theta(x, t; \eta_n^*)}, \quad \widehat D_n=B[\eta_n^*]+2\,i\,\theta'(x, t; \eta_n^*),
\end{gathered}
\end{gather}
and $\mu_{-j}$ and $\mu_{-j}'\,\, (j=1,2)$ satisfy
\begin{gather} \no
\begin{gathered}
\mu_{-1}(\eta_n^*)=\mathrm{e}^{iv_-\sigma_3}\begin{bmatrix}1\\0\end{bmatrix}
+\sum_{k=1}^{2N}C_k(\eta_n^*)\left[\mu_{-2}'(\eta_k)+\left(D_k+\frac{1}{\eta_n^*-\eta_k}\right)\mu_{-2}(\eta_k)\right]+\frac{1}{2\pi i}\int_\Sigma\frac{\left(M^+J\right)_1(\zeta)}{\zeta-\eta_n^*}\,\mathrm{d}\zeta,        \\
\mu_{-1}'(\eta_n^*)=-\sum_{k=1}^{2N}\frac{C_k(\eta_n^*)}{\eta_n^*-\eta_k}\left[\mu_{-2}'(\eta_k)+\left(D_k+\frac{2}{\eta_n^*-\eta_k}\right)
\mu_{-2}(\eta_k)\right]+\frac{1}{2\pi i}\int_\Sigma\frac{\left(M^+J\right)_1(\zeta)}{\left(\zeta-\eta_n^*\right)^2}\,\mathrm{d}\zeta, \qquad\qquad   \\
\mu_{-2}(\eta_k)=\mathrm{e}^{iv_-\sigma_3}\begin{bmatrix}0\\1 \end{bmatrix}
+\sum_{j=1}^{2N}\widehat C_j(\eta_k)\left[\mu_{-1}'(\eta_j^*)+\left(\widehat D_j+\frac{1}{\eta_k-\eta_j^*}\right)\mu_{-1}(\eta_j^*)\right]+\frac{1}{2\pi i}\int_\Sigma\frac{\left(M^+J\right)_2(\zeta)}{\zeta-\eta_k}\,\mathrm{d}\zeta,                    \\
\mu_{-2}'(\eta_k)=-\sum_{j=1}^{2N}\frac{\widehat C_j(\eta_k)}{\eta_k-\eta_j^*}\left[\mu_{-1}'(\eta_j^*)+\left(\widehat D_j+\frac{2}{\eta_k-\eta_j^*}\right)\mu_{-1}(\eta_j^*)\right]+\frac{1}{2\pi i}\int_\Sigma\frac{\left(M^+J\right)_2(\zeta)}{\left(\zeta-\eta_k\right)^2}\,\mathrm{d}\zeta, \qquad\qquad
\end{gathered}
\end{gather}
$\theta'(x, t; k)$ denotes the derivative of $\theta(x, t; k)$ with respect to variable $k$ and $\int_\Sigma$ the integral along the oriented contour shown in Fig. \ref{k-ZBC}.
\end{proposition}

\subsubsection{Reconstruction formula of the potential}

From the solution of the Riemann-Hilbert problem, one can obtain
\begin{gather}
M(x, t; k)=\mathrm{e}^{iv_-(x, t)\sigma_3}+\frac{M^{(1)}(x, t)}{k}+O\left(\frac{1}{k^2}\right), \quad \mathrm{as} \quad k\to\infty,
\end{gather}
where
\begin{gather}
\begin{aligned}
M^{(1)}(x, t)=&-\frac{1}{2\pi i}\int_\Sigma M^+(x, t; \zeta)\,J(x, t; \zeta)\,\mathrm{d}\zeta \qquad\qquad \\
& +\sum_{n=1}^{2N}\left[A[\eta_n]\,\mathrm{e}^{-2i\theta(\eta_n)}\left(\mu_{-2}'(\eta_n)+D_n\mu_{-2}(\eta_n)\right), A[\eta_n^*]\,\mathrm{e}^{2i\theta(\eta_n^*)}\left(\mu_{-1}'(\eta_n^*)+\widehat D_n\mu_{-1}(\eta_n^*)\right)\right].
\end{aligned}
\end{gather}
Substituting $\varPhi(x, t; k)=M(x, t; k)\,\mathrm{e}^{i\theta(x, t; k)\sigma_3}$ into Eq. (\ref{lax-x}) and matching $O\left(k\right)$ term, one obtains the reconstruction formula of the solution (potential) of the DNLS equation with ZBC and double poles as
\begin{gather}
q(x, t)=-2\,i\,\mathrm{e}^{iv_-(x, t)\sigma_3}\alpha^T\gamma,
\end{gather}
where the column vectors $\alpha=(\alpha^{(1)},\,\alpha^{(2)})^T$ and $\gamma=(\gamma^{(1)},\,\gamma^{(2)})^T$ are given by
\begin{gather*}
\begin{gathered}
 \alpha^{(1)}=\left(A[\eta_n^*]\,\mathrm{e}^{2i\theta(\eta_n^*)}\widehat D_n\right)_{2N\times 1},\quad \alpha^{(2)}=\left(A[\eta_n^*]\,\mathrm{e}^{2i\theta(\eta_n^*)}\right)_{2N\times 1},\v\\
 \gamma^{(1)}=\big(\mu_{-11}(\eta_n^*)\big)_{2N\times 1},\quad \gamma^{(2)}=\left(\mu_{-11}'(\eta_n^*)\right)_{2N\times 1}.
\end{gathered}
\end{gather*}

\subsubsection{Trace formula}

By using the asymptotic behavior of the scattering matrix as $k\to\infty$ and the Plemelj's formulae, $s_{11}(k)$ and $s_{22}(k)$ can be represented by the discrete spectrum and reflection coefficients. That is called the trace formula, which  is written as
\begin{gather}
\begin{gathered}
s_{11}(k)=\exp\left(-\frac{1}{2\pi i}\int_\Sigma\frac{\log\left[1-\rho(\zeta)\,\tilde\rho(\zeta)\right]}{\zeta-k}\,\mathrm{d}\zeta\right)\prod_{n=1}^{2N}
\left(\frac{k-\eta_n}{k-\widehat\eta_n}\right)^2\mathrm{e}^{-iv_0},\\[0.05in]
\no s_{22}(k)=\exp\left(\frac{1}{2\pi i}\int_\Sigma\frac{\log\left[1-\rho(\zeta)\,\tilde\rho(\zeta)\right]}{\zeta-k}\,\mathrm{d}\zeta\right)\prod_{n=1}^{2N}\left(\frac{k-\widehat\eta_n}{k-\eta_n}\right)^2
\mathrm{e}^{iv_0}.
\end{gathered}
\end{gather}

\subsubsection{Reflectionless potential: double-pole solitons}

We consider a special kind of solutions, reflectionless potential. From the Jost integral equation, one obtains $\varPhi_{\pm}(x, t; 0)=\mu_{\pm}(x, t; 0)=I$. Thus, $s_{11}(0)=1$. Combining the trace formula, one obtains that there exists an integer $j\in\mathbb{Z}$ such that
\begin{gather}
v_0=8\sum_{n=1}^{N}\mathrm{arg}(k_n)+2j\,\pi.
\end{gather}
From the reconstruction formula, the reflectionless potential is deduced by determinants:
\begin{gather}\label{wufanshe0-1}
q(x, t)=2i\,\frac{\mathrm{det}\left(G\right)}
{\mathrm{det}\left(I-H\right)}\,\mathrm{e}^{2iv_-(x, t)}, \quad G=\begin{bmatrix}
I-H&\beta\\[0.02in] \alpha^T&0\end{bmatrix},
\end{gather}
where the $4N\times 4N$ matrix $H=\begin{bmatrix}
H^{(1, 1)}& H^{(1, 2)}\\[0.02in]
H^{(2, 1)}& H^{(2, 2)}
\end{bmatrix}$ with $H^{(i, m)}=\left(H^{(i,  m)}_{n,  j}\right)_{2N\times 2N}\, (i, m=1, 2)$ is given by
\begin{gather*}
\begin{aligned}
&H^{(1, 1)}_{n, j}=\sum_{k=1}^{2N}C_k(\eta_n^*)\,\widehat C_j(\eta_k)\left[-\frac{1}{\eta_k-\eta_j^*}\left(\widehat D_j+\frac{2}{\eta_k-\eta_j^*}\right)+\left(D_k+\frac{1}{\eta_n^*-\eta_k}\right)\left(\widehat D_j+\frac{1}{\eta_k-\eta_j^*}\right)\right],  \\
&H^{(1, 2)}_{n, j}=\sum_{k=1}^{2N}C_k(\eta_n^*)\,\widehat C_j(\eta_k)\left[-\frac{1}{\eta_k-\eta_j^*}+\left(D_k+\frac{1}{\eta_n^*-\eta_k}\right)\right],  \\
&H^{(2, 1)}_{n, j}=\sum_{k=1}^{2N}C_k(\eta_n^*)\,\widehat C_j(\eta_k)\left[\frac{1}{\eta_k-\eta_j^*}\left(\widehat D_j+\frac{2}{\eta_k-\eta_j^*}\right)-\left(D_k+\frac{2}{\eta_n^*-\eta_k}\right)\left(\widehat D_j+\frac{1}{\eta_k-\eta_j^*}\right)\right],  \\
&H^{(2, 2)}_{n, j}=\sum_{k=1}^{2N}C_k(\eta_n^*)\,\widehat C_j(\eta_k)\left[\frac{1}{\eta_k-\eta_j^*}-\left(D_k+\frac{2}{\eta_n^*-\eta_k}\right)\right],
\end{aligned}
\end{gather*}
and $4N$ column vector $\beta=(\beta^{(1)},\,\beta^{(2)})^T$ with $\beta^{(1)}=\left(1\right)_{2N\times 1}, \, \beta^{(2)}=\left(0\right)_{2N\times 1}$.

However, this formula for the reflectionless potential is implicit since $v_-(x, t)$ is included. One needs to derive an explicit form for the reflectionless potential. From the trace formula and Jost integrable equation, one can deduce that
\begin{gather}
M(x, t; k)=I+k\sum_{n=1}^{2N}\left(\frac{\mathop\mathrm{P_{-2}}\limits_{k=\eta_n}\left[M(k)/k\right]}{\left(k-\eta_n\right)^2}+\frac{\mathop\mathrm{Res}\limits_{k=\eta_n}\left[M(k)/k\right]}{k-\eta_n}+\frac{\mathop\mathrm{P_{-2}}\limits_{k=\eta_n^*}\left[M(k)/k\right]}{\left(k-\eta_n^*\right)^2}+\frac{\mathop\mathrm{Res}\limits_{k=\eta_n^*}\left[M(k)/k\right]}{k-\eta_n^*}\right),
\end{gather}
which can  yield $\gamma$ explicitly.  Then substituting $\gamma$ into the reconstruction formula of the potential, one obtains
\begin{gather}\label{wufanshe0-2}
q(x, t)=2i\,\frac{\mathrm{det}\left(\widehat G\right)}
{\mathrm{det}(I-\widehat H)}\,\mathrm{e}^{iv_-(x, t)}, \quad \widehat G=\begin{bmatrix}
I-\widehat H&\beta\\[0.02in] \alpha^T&0\end{bmatrix},
\end{gather}
where the $4N\times 4N$ matrix $\widehat H=\begin{bmatrix}
\widehat H^{(1, 1)}& \widehat H^{(1, 2)}\\[0.02in]
\widehat H^{(2, 1)}& \widehat H^{(2, 2)}
\end{bmatrix}$ with $\widehat H^{(i, m)}=\left(\widehat H^{(i,  m)}_{n,  j}\right)_{2N\times 2N}\, (i, m=1, 2)$ is given by
\begin{gather*}
\begin{aligned}
&\widehat H^{(1, 1)}_{n, j}\!\!=\eta_n^*\sum_{k=1}^{2N}\frac{C_k(\eta_n^*)\,\widehat C_j(\eta_k)}{\eta_k}\!\left[-\frac{1}{\eta_k-\eta_j^*}\left(\widehat D_j+\frac{2}{\eta_k-\eta_j^*}\right)+\frac{\eta_k}{\eta_j^*}\!\left(D_k+\frac{1}{\eta_n^*-\eta_k}-\frac{1}{\eta_k}\right)\!\left(\widehat D_j+\frac{1}{\eta_k-\eta_j^*}-\frac{1}{\eta_j^*}\!\right)\!\right],  \\
&\widehat H^{(1, 2)}_{n, j}=\eta_n^*\sum_{k=1}^{2N}\frac{C_k(\eta_n^*)\,\widehat C_j(\eta_k)}{\eta_k}\left[-\frac{1}{\eta_k-\eta_j^*}+\frac{\eta_k}{\eta_j^*}\left(D_k+\frac{1}{\eta_n^*-\eta_k}-\frac{1}{\eta_k}\right)\right],  \\
&\widehat H^{(2, 1)}_{n, j}=\sum_{k=1}^{2N}\frac{C_k(\eta_n^*)\,\widehat C_j(\eta_k)}{\eta_n^*-\eta_k}\left[\frac{1}{\eta_k-\eta_j^*}\left(\widehat D_j+\frac{2}{\eta_k-\eta_j^*}\right)-\frac{\eta_k}{\eta_j^*}\left(D_k+\frac{2}{\eta_n^*-\eta_k}\right)\left(\widehat D_j+\frac{1}{\eta_k-\eta_j^*}-\frac{1}{\eta_j^*}\right)\right],  \\
&\widehat H^{(2, 2)}_{n, j}=\sum_{k=1}^{2N}\frac{C_k(\eta_n^*)\,\widehat C_j(\eta_k)}{\eta_n^*-\eta_k}\left[\frac{1}{\eta_k-\eta_j^*}-\frac{\eta_k}{\eta_j^*}\left(D_k+\frac{2}{\eta_n^*-\eta_k}\right)\right].
\end{aligned}
\end{gather*}

\begin{figure}[!t]
\centering
\includegraphics[scale=0.6]{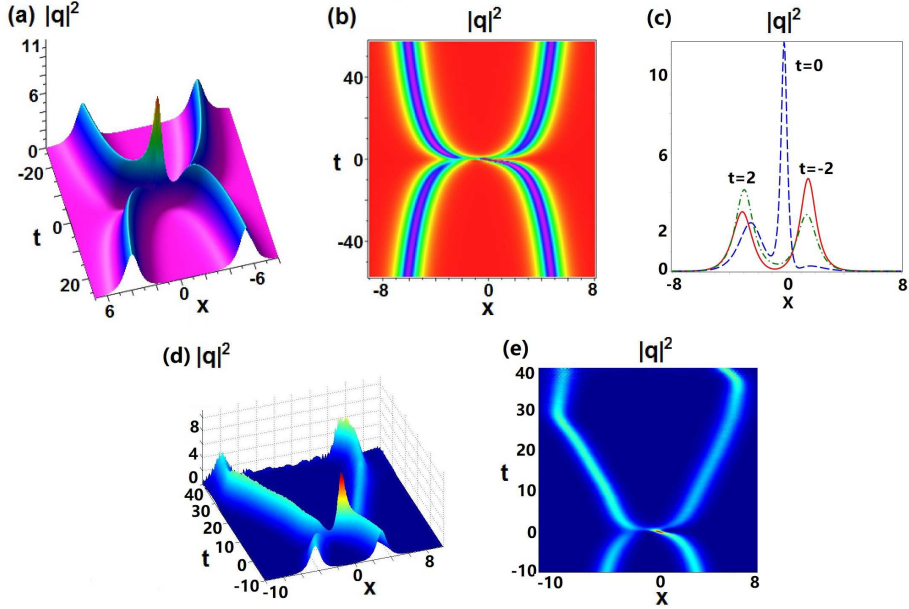}
\vspace{-0.1in}\caption{Double-pole soliton solution of DNLS equation (\ref{DNLS}) with ZBC. (a) 3D profile; (b) intensity profile; (c) profiles for $t=-2$ (solid line), $t=0$ (dashed line), and $t=2$ (dash-dot line); (d)-(e) time evolution of double-pole soliton. Parameters are $N=1, k_1=(1+i)/2, A[k_1]=B[k_1]=1$.}
\label{F-0}
\end{figure}

Combining Eqs.~(\ref{wufanshe0-1}) and (\ref{wufanshe0-2}), we have the following Theorem:

\begin{theorem} The explicit double-pole soliton solutions of DNLS Eq.~(\ref{DNLS}) with ZBC are found by
\begin{gather}
q(x, t)=2i\left(\frac{\mathrm{det}\left(\widehat G\right)}{\mathrm{det}\left(I-\widehat H\right)}\right)^2\frac{\mathrm{det}\left(I-H\right)}{\mathrm{det}\left(G\right)}.
\end{gather}
\end{theorem}

For example, we have the single double-pole solution of Eq.~(\ref{DNLS}) with ZBC for parameters $N=1, k_1=\frac{1}{2}\left(1+i\right), A[k_1]=1, B[k_1]=1$ as $q(x,t)=P_1/P_2$ with
\bee
\begin{array}{rl}
P_1=&\!\!\! \frac{i}{4}\Big\{[4(i-1)t^3+[1-2i-2(1+i)x]t^2+[(13i-11)/8+(i-1)x^2+(2-i)x]t+(10-i)/32 \v\\
    &\!\!\!-(1+i)x^3/2+(3+2i)x^2/4-(7+i)x/16]e^{4x}+[(1+i)t^3+((1-i)x-3)t^2/4 \v\\
    &\!\!\!+[(11+5i)/32+(1+i)x^2/4+ix/4]t-7i/128+(1-i)x^3/8-x^2/16-(1+i)x/64]e^{2x} \v\\
   &\!\!\!+[i/2-1/4-(1+i)t+(1-i)x/2]e^{6x}+(1-i)t/64-[1+2(1+i)x]/256\Big\}^2e^{x+it}, \v\\
P_2=&\!\!\! \Big\{[t^4-t^3/2+(2x^2-x+2)t^2/4-(2x^2-x+7)t/16+(x^2-x/2+5/8)(x^2-x/2+1/8)/16]e^{4x} \v\\
  &\!\!\! +(1-8t)e^{2x}/512+(8t-1)e^{6x}/32+e^{8x}/16+1/4096\Big\}\Big\{[4(i-1)t^3+(5+2i-2(1+i)x)t^2\v\\
  &\!\!\! +[(i-1)x^2+ix-(3+11i)/8]t+(2+7i)/32-(1+i)x^3/2+(3+2i)x^2/4-(7+9i)x/16]e^{4x} \v\\
  &\!\!\! +[(1+i)t^3+(1-4i+2(1-i)x)t^2/4+(5(i-1)+8(1+i)x^2-8(2+i)x)t/32+i/128 \v\\
  &\!\!\!   +(1-i)x^3/8-x^2/16+(7-i)x/64]e^{2x}+[(1+i)t+(1-2i)/4+(i-1)x/2]e^{6x}\v\\
  &\!\!\! +(i-1)t/64+[1+2(1+i)x]/256\Big\}.
\end{array}
\ene
Fig.~\ref{F-0} exhibits the dynamical structure of the double-pole soliton. Figs.~\ref{F-0}(a)-(c) exhibit the exact double-pole soliton, which is equivalent to the elastic collisions of two bright-bright solitons. Moreover, we use the exact double-pole soliton at $t=-10$ as the initial condition without a small noise to numerically test the wave propagation such that we find that the wave stably propagates in a short time (e.g., $t\in [-10, 10]$), but after $t>10$, slowly separate, and then from $t=30$, the two bright waves begin to close, that is, the double-pole soliton is unstable even if there is no noise (see Figs.~\ref{F-0}(d)-(e)).

\section{IST with NZBCs and double poles}

Recently, the ISTs of integrable nonlinear systems with NZBCs have attracted more and more attention~\cite{Prinari2006,Demontis2013, Demontis2014,Biondini2014,Prinari2015a, Pichler2017,zhang18}. In this section, we will search a double-pole solution $q(x, t)$ for the DNLS equation (\ref{DNLS}) with $\sigma=-1$ and the NZBCs
\begin{gather}\label{NZBC}
q(x, t)\to q_{\pm},\quad \mathrm{as}\quad x\to \pm\infty,
\end{gather}
by the IST, where $\left|q_{\pm}\right|=q_0\ne 0$.  The ISTs for DNLS equation (\ref{DNLS}) with NZBCs (\ref{NZBC}) were also studied~\cite{Kawata1978, Chen2004, Chen2006, Lashkin2007}, but they only considered the case of simple poles by solving the corresponding Gel'fand-Levitan-Marchenko integral equations. In this section, we try to present the IST with NZBCs and double poles based on another approach, that is, the Riemann-Hilbert problem.

\subsection{Direct scattering with NZBCs and double poles}

\subsubsection{Jost solutions, analyticity, and continuity}

Let $x\to \pm\infty$, we consider the asymptotic scattering problem of the Lax pair (\ref{lax-x}, \ref{lax-t}):
\begin{align} \label{laxn}
\begin{aligned}
\varPhi_x&=X_{\pm}\varPhi, & X_{\pm}&=k\left(ik\sigma_3+Q_{\pm}\right), \\[0.04in]
\varPhi_t&=T_{\pm}\varPhi, &T_{\pm}&=-\left(2k^2-q_0^2\right)X_{\pm},
\end{aligned}
\end{align}
where
\begin{gather*}
Q_{\pm}=
\begin{bmatrix}
0&q_{\pm}\\
-q_{\pm}^*&0
\end{bmatrix}.
\end{gather*}

One can obtain the fundamental matrix solution of Eq.~(\ref{laxn}) as
\begin{align}
\varPhi^{bg}(x, t; k)=
\left\{
\begin{alignedat}{2}
&E_{\pm}(k)\,\mathrm{e}^{i\theta(x, t; k)\sigma_3}, &&k\ne\pm iq_0,\\[0.04in]
&I+\left(x+3\,q_0^2\,t\right)X_{\pm}(k),&\quad&k=\pm iq_0,
\end{alignedat}\right.
\end{align}
where
\begin{align} \label{lp}
E_{\pm}(k)=
\begin{bmatrix}
1&\frac{iq_{\pm}}{k+\lambda} \\[0.05in]
\frac{iq_{\pm}^*}{k+\lambda}&1
\end{bmatrix},\quad
\theta(x, t; k)=k \lambda\left[x-\left(2k^2-q_0^2\right)t\right], \quad \lambda=\sqrt{k^2+q_0^2}.
\end{align}

Since $\lambda(k)$ stands for a two-sheeted Riemann surface, for convenience, taking a uniformization variable:
\begin{align}
z=k+\lambda,
\end{align}
which was first introduced in \cite{Faddeev1987}, we will illustrate the scattering problem on a standard $z$-plane instead of the two-sheeted Riemann surface by the inverse mapping:
\begin{align} \no
k=\frac{1}{2}\left(z-\frac{q_0^2}{z}\right),\quad \lambda=\frac{1}{2}\left(z+\frac{q_0^2}{z}\right).
\end{align}

Define $\Sigma$, $D^+$, and $D^-$ on $z$-plane as
\begin{gather} \no
\Sigma:=\mathbb{R}\cup i\mathbb{R}\backslash\left\{0\right\}, \quad D^+:=\left\{z: \mathrm{Re}\,z\,\mathrm{Im}\,z>0\right\}, \quad D^-:=\left\{z: \mathrm{Re}\,z\,\mathrm{Im}\,z<0\right\}.
\end{gather}
From the mapping relation between $k$-plane and $z$-plane under the uniformization variable, one finds that
\begin{gather}\no
\mathrm{Im}\,(k\lambda)=0 \quad \mathrm{as} \quad z\in\Sigma ; \quad \mathrm{Im}\,(k\lambda)>0\quad\mathrm{as}\quad z\in D^+;\quad  \mathrm{Im}\,(k\lambda)<0\quad\mathrm{as}\quad z\in D^-.
\end{gather}

\begin{figure}[!t]
\centering
\includegraphics[scale=0.35]{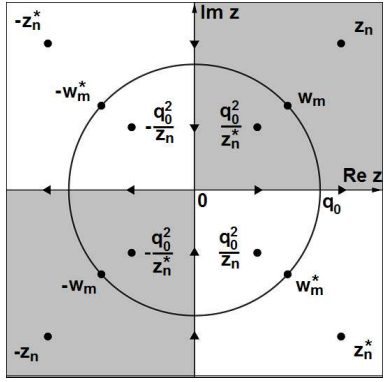}
\caption{Complex $z$-plane, showing the region $D^+$ (grey region), the region $D^-$ (white region), the discrete spectrum and the orientation of the contours for the Riemann-Hilbert problem.}
\label{z-NZBC}
\end{figure}

According to the IST technique, one needs to define the Jost solutions $\varPhi_{\pm}(x, t; z)$ such that
\begin{align}\label{NZBC-Jxjianjin}
\varPhi_{\pm}(x, t; z)=E_{\pm}(z)\,\mathrm{e}^{i\theta(x, t; z)\sigma_3}+o\left(1\right),\quad z\in\Sigma,\quad \mathrm{as}\,\, x\to\pm\infty,
\end{align}
and the modified Jost solutions $\mu_{\pm}(x, t; z)$ via dividing by the asymptotic exponential oscillations
\begin{align}
\mu_{\pm}(x, t; z)=\varPhi_{\pm}(x, t; z)\,\mathrm{e}^{-i\theta(x, t; z)\sigma_3},
\end{align}
such that $\lim_{x\to\pm\infty}\mu_{\pm}(x, t; z)=E_{\pm}(z).$

It follows from Eq.~(\ref{lax-x}) that the modified Jost solutions $\mu_{\pm}(x, t; z)$ satisfy the Jost integral equations
\begin{align}\label{Jost-int}
\mu_{\pm}(x, t; z)=E_{\pm}(z)+\left\{
\begin{aligned}
&k(z)\int_{\pm\infty}^xE_{\pm}(z)\,\mathrm{e}^{ik(z)\lambda(z)(x-y)\widehat\sigma_3}\left[{E_{\pm}^{-1}(z)}\Delta Q_{\pm}(y, t)\,\mu_{\pm}(y, t; z)\right]\,\mathrm{d}y, && z\ne\pm iq_0,\\[0.05in]
&k(z)\int_{\pm\infty}^x\left[I+\left(x-y\right)X_{\pm}(z)\right]\Delta Q_{\pm}(y, t)\,\mu_{\pm}(y, t; z)\,\mathrm{d}y, &&  z=\pm iq_0,
\end{aligned}\right.
\end{align}
which is used to deduce the following analyticity of the (modified) Jost solution, where $e^{\alpha\widehat \sigma_3}[\bm \cdot]:=e^{\alpha\sigma_3}[\bm\cdot]e^{-\alpha\sigma_3}$, $\Delta Q_{\pm}(x, t):=Q(x, t)-Q_{\pm}$.
\begin{proposition}
Suppose $\left(1+\left|x\right|\right)\left(q(x,t)-q_{\pm}\right)\in L^1\left(\mathbb{R^{\pm}}\right)$, then $\varPhi_\pm(x, t; z)$ have the following  properties:
\begin{itemize}
\item Eq.~(\ref{lax-x}) has the unique solution $\varPhi_{\pm}(x, t; z)$ satisfying Eq.~(\ref{NZBC-Jxjianjin}) on $\Sigma$,

\item $\varPhi_{+1, -2}(x, t; z)$  can be analytically extended  to $D^{+}$ and continuously extended to $D^{+}\cup\Sigma$,

\item $\varPhi_{-1, +2}(x, t; z)$ can be analytically extend  to $D^{-}$ and continuously extended to $D^-\cup\Sigma$.
\end{itemize}
\end{proposition}
The analyticity and continuity for $\mu_{\pm}(x, t; k)$ follow trivially from those of $\varPhi_{\pm}(x, t; k)$. In the same manner as the case of ZBC in Sec. 2, one can also confirm the Jost solution $\varPhi_{\pm}(x, t; z)$ the simultaneous solution for both parts of the Lax pair (\ref{lax-x},  \ref{lax-t}).

\subsubsection{Scattering matrix, analyticity,  and continuity}

Liouville's formulae implies $\varPhi_{\pm}$ are the fundamental matrix solutions in $z\in \Sigma\backslash\left\{\pm iq_0\right\}$, so one can define the scattering $S(z)=(s_{ij}(x))_{2\times 2}$, which does not depend on variable $x$ and $t$:
\begin{align}
\varPhi_+(x, t; z)=\varPhi_-(x, t; z)\,S(z),\quad z\in \Sigma\backslash\left\{\pm iq_0\right\}.
\end{align}
Then one has the Wronskian representations for the scattering coefficients:
\begin{align}\label{S-lie}
\begin{aligned}
s_{11}(z)=\frac{{\rm Wr}(\varPhi_{+1}(x, t; z), \varPhi_{-2}(x, t; z))}{s(z)},\quad s_{22}(z)=\frac{{\rm Wr}(\varPhi_{-1}(x, t; z), \varPhi_{+2}(x, t; z))}{s(z)},\\[0.05in]
s_{12}(z)=\frac{{\rm Wr}(\varPhi_{+2}(x, t; z), \varPhi_{-2}(x, t; z))}{s(z)},\quad s_{21}(z)=\frac{{\rm Wr}(\varPhi_{-1}(x, t; z), \varPhi_{+1}(x, t; z))}{s(z)},
\end{aligned}
\end{align}
where ${\rm Wr}(\bm\cdot, \bm\cdot)$ denotes the Wronskian determinant and $s(z):=1+\frac{q_0^2}{z^2}$. From these Wronskian representations, one can extend the analytical regions of $s_{11}(z)$ and $s_{22}(z)$.
\begin{proposition}
Suppose $q(x,t)-q_{\pm}\in L^1\left(\mathbb{R^{\pm}}\right)$. Then $s_{11}(z)$ can be analytically extended  to $D^+$ and continuously extended to $D^+\cup\Sigma\backslash\{\pm iq_0\}$, while $s_{22}(z)$ can be analytically extended  to $D^-$ and continuously extended to $D^-\cup\Sigma\backslash\{\pm iq_0\}$. Moreover, both $s_{12}(z)$ and $s_{21}(z)$ are continuous in $z\in \Sigma\backslash\left\{\pm iq_0\right\}$.
\end{proposition}
\begin{proposition}
Suppose $\left(1+\left|x\right|\right)\left(q-q_{\pm}\right)\in L^1\left(\mathbb{R^{\pm}}\right)$. Then $\lambda(z)\,s_{11}(z)$ can be extend analytically to $D^+$ and continuously extended to $D^+\cup\Sigma$ while $\lambda(z)\,s_{22}(z)$ can be analytically extended  to $D^-$ and continuously extended to $D^-\cup\Sigma$. Moreover, both $\lambda(z)\,s_{12}(z)$ and $\lambda(z)\,s_{21}(z)$ are continuous in $\Sigma$.
\end{proposition}
To solve the Riemann-Hilbert problem in the inverse process, we restrict our consideration to the potential without spectral singularities \cite{Zhou1989}, i.e., $s_{11}(z)s_{22}(z)\ne 0$ in $\Sigma$, such that $1/s_{11}(z)$ and $1/s_{22}(z)$ can extend continuously to $\Sigma$.  The reflection coefficients $\rho(z)$ and $\tilde\rho(z)$ are defined as
\begin{align}\label{fanshe}
\rho(z)=\frac{s_{21}(z)}{s_{11}(z)},\quad
\tilde\rho(z)=\frac{s_{12}(z)}{s_{22}(z)}, \quad z\in\Sigma\backslash\left\{\pm iq_0\right\}.
\end{align}

\subsubsection{Symmetry structures}

The symmetries of $X(x, t; z)$, $T(x, t; z)$, Jost solutions, scattering matrix and reflection coefficients in the case of NZBCs are more complicated than ones in the ZBC.

\begin{proposition}[Symmetry reductions~\cite{Chen2004}] For the case of NZBCs, $X(x, t; z)$ and $T(x, t; z)$ in the Lax pair (\ref{lax-x}, \ref{lax-t}), Jost solutions, scattering matrix and reflection coefficients admit the following three reduction conditions on $z$-plane:
\begin{itemize}
\item The first symmetry reduction
\begin{gather} \no
\begin{gathered}
X(x, t; z)=\sigma_2\,X(x, t; z^*)^*\,\sigma_2, \quad T(x, t; z)=\sigma_2\,T(x, t; z^*)^*\,\sigma_2,\v \\
\varPhi_{\pm}(x, t; z)=\sigma_2\,\varPhi_{\pm}(x, t; z^*)^*\,\sigma_2, \quad S(z)=\sigma_2\,S(z^*)^*\,\sigma_2, \quad \rho(z)=-\tilde\rho(z^*)^*.
\end{gathered}
\end{gather}
\item The second symmetry reduction
\begin{gather} \no
\begin{gathered}
X(x, t; z)=\sigma_1\,X(x, t; -z^*)^*\,\sigma_1,\quad T(x, t; z)=\sigma_1\,T(x, t; -z^*)^*\,\sigma_1,\v\\
\varPhi_{\pm}(x, t; z)=\sigma_1\,\varPhi_{\pm}(x, t; -z^*)^*\sigma_1,\quad S(z)=\sigma_1\,S(-z^*)^*\sigma_1, \quad \rho(z)=\tilde\rho(-z^*)^*.
\end{gathered}
\end{gather}
\item The third symmetry reduction
\begin{gather} \no
\begin{gathered}
X(x, t; z)=X\left(x, t; -\frac{q_0^2}{z}\right),\quad T(x, t; z)=T\left(x, t; -\frac{q_0^2}{z}\right),\v\\
\varPhi_{\pm}(x, t; z)=\frac{i}{z}\,\varPhi_{\pm}\left(x, t; -\frac{q_0^2}{z}\right)\sigma_3Q_{\pm}, \quad S(z)=\left(\sigma_3Q_-\right)^{-1}S\left(-\frac{q_0^2}{z}\right)\sigma_3Q_+,\quad \rho(z)=\frac{q_-^*}{q_-}\,\tilde\rho\left(-\frac{q_0^2}{z}\right).
\end{gathered}
\end{gather}
\end{itemize}
\end{proposition}

\subsubsection{Discrete spectrum with double poles}

The previous works on DNLS equation with NZBCs focused on the single poles of the scattering coefficients~\cite{Kawata1978, Chen2004, Chen2006, Lashkin2007}. We here consider the case of $s_{11}(z)$ with {\it double zeros}. We suppose that $s_{11}(z)$ has $N_1$ {\it double zeros} in $\left\{z\in\mathbb{C}\,|\,\mathrm{Re} \,z>0, \mathrm{Im}\, z>0, \left|z\right|>q_0\right\}$ denoted by $z_n$, and $N_2$ {\it double zeros} in $\left\{z=q_0\,\mathrm{e}^{i\phi}: 0<\phi<\frac{\pi}{2}\right\}$ denoted by $w_m$, that is, $s_{11}(z_0)=s'_{11}(z_0)=0$ and $s''_{11}(z_0)\not=0$ if $z_0$ is a double zero of $s_{11}(z)$. From the symmetries of the scattering matrix, the discrete spectrum is the set
\begin{align}
Z=\left\{\pm z_n,\, \pm z_n^*, \, \pm\frac{q_0^2}{z_n},\, \pm\frac{q_0^2}{z_n^*}\right\}_{n=1}^{N_1}
\bigcup\Big\{\pm w_m, \,\pm w_m^*\Big\}_{m=1}^{N_2},
\end{align}
whose distributions are displayed in Fig. \ref{z-NZBC}.

For convenience, let
\begin{align}
b[z_0]=\left\{
\begin{aligned}
\frac{\varPhi_{+1}(x, t; z_0)}{\varPhi_{-2}(x, t; z_0)}, \quad z_0\in Z\cap D^+,\\[0.05in]
\frac{\varPhi_{+2}(x, t; z_0)}{\varPhi_{-1}(x, t; z_0)}, \quad z_0\in Z\cap D^-,
\end{aligned}\right.\qquad
d[z_0]=\left\{
\begin{aligned}
\frac{\varPhi_{+1}'(x, t; z_0)-b[z_0]\,\varPhi_{-2}'(x, t; z_0)}{\varPhi_{-2}(x, t; z_0)},\quad z_0\in Z\cap D^+,\\[0.05in]
\frac{\varPhi_{+2}'(x, t; z_0)-b[z_0]\,\varPhi_{-1}'(x, t; z_0)}{\varPhi_{-1}(x, t; z_0)},\quad z_0\in Z\cap D^-.
\end{aligned}\right.
\end{align}
\begin{align}
A[z_0]=\left\{
\begin{aligned}
\frac{2\,b[z_0]}{s_{11}''(z_0)},\quad z_0\in Z\cap D^+,\\[0.05in]
\frac{2\,b[z_0]}{s_{22}''(z_0)},\quad z_0\in Z\cap D^-,
\end{aligned}\right. \qquad
B[z_0]=\left\{
\begin{aligned}
\frac{d[z_0]}{b[z_0]}-\frac{s_{11}'''(z_0)}{3\,s_{11}''(z_0)},\quad z_0\in Z\cap D^+,\\[0.05in]
\frac{d[z_0]}{b[z_0]}-\frac{s_{22}'''(z_0)}{3\,s_{22}''(z_0)},\quad z_0\in Z\cap D^-.
\end{aligned}\right.
\end{align}
where $\frac{\bm\cdot}{\bm\cdot}$ denotes the proportionality coefficient. Then one can pose the following compact forms:
\begin{align}\label{erjieliu}
\begin{aligned}
\mathop\mathrm{P_{-2}}\limits_{z=z_0}\left[\frac{\varPhi_{+1}(x, t; z)}{s_{11}(z)}\right]&=A[z_0]\,\varPhi_{-2}(x, t; z_0), \quad z_0\in Z\cap D^+,\\[0.05in]
\mathop\mathrm{P_{-2}}\limits_{z=z_0}\left[\frac{\varPhi_{+2}(x, t; z)}{s_{22}(z)}\right]&=A[z_0]\,\varPhi_{-1}(x, t; z_0), \quad z_0\in Z\cap D^-,\\[0.05in]
\mathop\mathrm{Res}\limits_{z=z_0}\left[\frac{\varPhi_{+1}(x, t; z)}{s_{11}(z)}\right]&=A[z_0]\left[\varPhi_{-2}'(x, t; z_0)+B[z_0]\, \varPhi_{-2}(x, t; z_0)\right], \quad z_0\in Z\cap D^+,\\[0.05in]
\mathop\mathrm{Res}\limits_{z=z_0}\left[\frac{\varPhi_{+2}(x, t; z)}{s_{22}(z)}\right]&=A[z_0]\left[\varPhi_{-1}'(x, t; z_0)+B[z_0]\, \varPhi_{-1}(x, t; z_0)\right], \quad z_0\in Z\cap D^-.
\end{aligned}
\end{align}
where $\mathop\mathrm{P_{-2}}\limits_{z=z_0}\left[\bm\cdot\right]$ stands for the coefficient of $O\left(\frac{1}{\left(z-z_0\right)^2}\right)$ term in the Laurent expansion of $\bm\cdot$ at $z=z_0$.

\begin{proposition}
Given $z_0\in Z$, three symmetry relations for $A[z_0]$ and $B[z_0]$ are given as
\begin{itemize}
\item The first symmetry relation
\begin{align} \no
A[z_0]=-A[z_0^*]^*, \quad B[z_0]=B[z_0^*]^*.
\end{align}
\item The second symmetry relation
\begin{align}\no
A[z_0]=A[-z_0^*]^*, \quad B[z_0]=-B[-z_0^*]^*.
\end{align}
\item The third symmetry relation
\begin{align}\no
A[z_0]=\frac{z_0^4\,q_-^*}{q_0^4\,q_-}\,A\left[-\frac{q_0^2}{z_0}\right],\quad B[z_0]=\frac{q_0^2}{z_0^2}\,B\left[-\frac{q_0^2}{z_0}\right]+\frac{2}{z_0}.
\end{align}
\end{itemize}
\end{proposition}
\begin{proposition}
For $n=1, 2, \cdots, N_1$ and $m=1, 2, \cdots, N_2$, one has
\begin{gather}\begin{gathered}
\begin{aligned}
&A[z_n]=-A[z_n^*]^*=-A[-z_n]=A[-z_n^*]^*  \qquad\qquad\qquad\\[0.05in]
&\qquad\,\, =\frac{z_n^4q_-^*}{q_0^4q_-}\,A\left[-\frac{q_0^2}{z_n}\right]=\frac{z_n^4q_-^*}{q_0^4q_-}\,A\left[\frac{q_0^2}{z_n^*}\right]^*
=-\frac{z_n^4q_-^*}{q_0^4q_-}\,A\left[-\frac{q_0^2}{z_n^*}\right]^*=-\frac{z_n^4q_-^*}{q_0^4q_-}\,A\left[\frac{q_0^2}{z_n}\right],
\end{aligned} \\[0.05in]
 A[w_m]=-A[w_m^*]^*=A[-w_m^*]^*=-A[-w_m]=\frac{w_m^4q_-^*}{q_0^4q_-}A[w_m]^*,\qquad\qquad\qquad\qquad\qquad
  \\[0.05in]
\begin{aligned}
B[z_n]=B[z_n^*]^*=&-B[-z_n^*]=-B[-z_n]=\frac{q_0^2}{z_n^2}\,B\left[-\frac{q_0^2}{z_n}\right]+\frac{2}{z_n}
        =\frac{q_0^2}{z_n^2}\,B\left[-\frac{q_0^2}{z_n^*}\right]^*+\frac{2}{z_n}\\[0.05in]
={}&-\frac{q_0^2}{z_n^2}\,B\left[\frac{q_0^2}{z_n^*}\right]^*+\frac{2}{z_n}=-\frac{q_0^2}{z_n^2}\,B\left[\frac{q_0^2}{z_n}\right]+\frac{2}{z_n},
\end{aligned}\\[0.05in]
B[w_m]=B[w_m^*]^*=-B[-w_m^*]^*=-B[-w_m]=-\frac{q_0^2}{w_m^2}B[w_m]^*+\frac{2}{w_m}.\qquad\qquad\qquad\qquad
\end{gathered}\end{gather}
\end{proposition}

\subsubsection{Asymptotic behaviors}

To propose and solve the Riemann-Hilbert problem in the following inverse problem, one has to find the asymptotic behaviors of the modified Jost solutions and scattering matrix as $z\to\infty$ and $z\to 0$, which differ from the case of ZBC. The standard Wentzel-Kramers-Brillouin (WKB) expansions are used to derive the asymptotic behaviors of the modified Jost solutions. We consider the following ansatz for the expansions of the modified Jost solutions $\mu_{\pm}(x, t; z)$ as $z\to\infty$ and $z\to 0$
\begin{gather}
\begin{gathered}
\mu_{\pm}(x, t; z)=\sum_{j=0}^n\frac{\mu_{\pm}^{[j]}(x, t)}{z^j}+O\left(\frac{1}{z^{n+1}}\right), \quad \mathrm{as}\quad z\to\infty, \\[0.02in]
\mu_{\pm}(x, t; z)=\sum_{j=-1}^n\mu_{\pm}^{[j]}(x, t)\,z^{j}+O\left(z^{n+1}\right), \quad \mathrm{as}\quad z\to 0,
\end{gathered}
\end{gather}
and substituting $\varPhi_{\pm}(x, t; k)=\mu_{\pm}(x, t; z)\,\mathrm{e}^{i\theta(x, t; k)\sigma_3}$ with these expansions into Eq. (\ref{lax-x}). By matching the $O\left(z^2\right)$, $O\left(z\right)$ and $O\left(1\right)$ terms as $z\to\infty$, and $O\left(z^{-3}\right)$, $O\left(z^{-2}\right)$ and $O\left(z^{-1}\right)$ terms as $z\to 0$, and
combining with the asymptotic behaviors of the modified Jost soutions $\mu_{\pm}(x, t; k)$ as $x\to\pm\infty$, one deduces the asymptotic behaviors as $z\to\infty$ and $z\to 0$.
\begin{proposition}
The asymptotic behaviors for the modified Jost solutions are given as
\begin{gather}
\begin{gathered}
\mu_{\pm}(x, t; z)=\mathrm{e}^{iv_{\pm}(x, t)\sigma_3}+O\left(\frac{1}{z}\right),\quad z\to\infty,\qquad \\
\mu_{\pm}(x, t; z)=\frac{i}{z}\,\mathrm{e}^{iv_{\pm}(x, t)\sigma_3}\sigma_3\,Q_{\pm}+O\left(1\right),\quad z\to0,
\end{gathered}
\end{gather}
where
\begin{gather}
v_{\pm}(x, t)=\frac{1}{2}\int_{\pm\infty}^x\left(\left| q(y, t)\right|^2-q_0^2\right)\mathrm{d}y.
\end{gather}
\end{proposition}
The asymptotic behavior of the scattering matrix can be yielded by Wronskian presentations of scattering matrix and the asymptotic behaviors of the modified Jost solutions.
\begin{proposition}
The asymptotic behaviors of the scattering matrix are
\begin{gather}
\begin{gathered}
S(z)=\mathrm{e}^{-iv_0\sigma_3}+O\left(\frac{1}{z}\right),  \quad z\to\infty, \quad\qquad\qquad \\
S(z)=\mathrm{diag}\left(\frac{q_-}{q_+}, \, \frac{q_+}{q_-}\right)\mathrm{e}^{iv_0\sigma_3}+O\left(z\right),\quad z\to 0,
\end{gathered}
\end{gather}
where $v_0$ is a constant given by
\begin{gather}
v_0=\frac{1}{2}\int_{-\infty}^{+\infty}\left(\left| q(y, t)\right|^2-q_0^2\right)\,\mathrm{d}y.
\end{gather}
\end{proposition}

\subsection{Inverse problem with NZBCs and double poles}

\subsubsection{Riemann-Hilbert problem}

We present the following  Riemann-Hilbert problem for the case of NZBCs, whose form is similar to one of ZBC.

\begin{proposition}
Let the sectionally meromorphic matrices
\begin{gather}
M(x, t; z)=\left\{
\begin{aligned}
\left(\frac{\mu_{+1}(x, t; z)}{s_{11}(z)},\, \mu_{-2}(x, t; z)\right),\quad z\in D^+, \\[0.05in]
\left(\mu_{-1}(x, t; z),\, \frac{\mu_{+2}(x, t; z)}{s_{22}(z)}\right), \quad z\in D^-,
\end{aligned}\right.
\end{gather}
and
\begin{gather}
M^{\pm}(x, t; z)=\lim_{\begin{matrix}
z'\to z\\z'\in D^{\pm}
\end{matrix}}M(x, t; z'), \quad z\in\Sigma.
\end{gather}
Then a multiplicative matrix Riemann-Hilbert problem is proposed:
\begin{itemize}
\item Analyticity: $M(x, t; z)$ is analytic in $\left(D^+\cup D^-\right)\backslash Z$ and has double poles in $Z$.
\item Jump condition:
\begin{gather}
M^-(x, t; z)=M^+(x, t; z)\left(I-J(x, t; z)\right), \quad z\in\Sigma,
\end{gather}
where
\begin{gather*}
J(x, t; z)=\mathrm{e}^{i\theta(x, t; z)\hat\sigma_3}
\begin{bmatrix}
0&-\tilde\rho(z)\\[0.05in]
\rho(z)&\rho(z)\,\tilde\rho(z)
\end{bmatrix}.
\end{gather*}
\item Asymptotic behaviors:
\begin{gather}
\begin{gathered}
M(x, t; z)=\mathrm{e}^{iv_-(x, t)\sigma_3}+O\left(\frac{1}{z}\right), \quad z\to\infty,\qquad \\
M(x, t; z)=\frac{i}{z}\,\mathrm{e}^{iv_-(x, t)\sigma_3}\sigma_3\,Q_-+O\left(1\right), \quad z\to 0.
\end{gathered}
\end{gather}
\end{itemize}
\end{proposition}
To solve  the above-mentioned Riemann-Hilbert problem conveniently, we define $\widehat\eta_n=-\frac{q_0^2}{\eta_n}$ with
\begin{align}
\eta_n=\left\{
\begin{aligned}
&z_n, & n&=1, 2, \cdots, N_1,\\[0.05in]
-{}&z_{n-N_1}, &n&=N_1+1, N_1+2, \cdots, 2N_1,\\[0.05in]
&\frac{q_0^2}{z^*_{n-2N_1}}, &n&=2N_1+1, 2N_1+2, \cdots, 3N_1,\\[0.05in]
-{}&\frac{q_0^2}{z^*_{n-3N_1}},&n&=3N_1+1,3N_1+2, \cdots, 4N_1,\\[0.05in]
&w_{n-4N_1},& n&=4N_1+1, 4N_1+2, \cdots, 4N_1+N_2,\\[0.05in]
-{}&w_{n-4N_1-N_2},&n&=4N_1+N_2+1, 4N_1+N_2+2, \cdots, 4N_1+2N_2.
\end{aligned}\right.
\end{align}
\begin{proposition}
The solution of the Riemann-Hilbert problem with double poles is given by
\begin{gather}
\begin{gathered}
M(x, t; z)=\mathrm{e}^{iv_-(x, t)\sigma_3}\left(I+\frac{i}{z}\,\sigma_3\,Q_-\right)+\frac{1}{2\pi i}\int_\Sigma\frac{(M^+J)(x, t; \zeta)}{\zeta-z}\,\mathrm{d}\zeta \qquad\qquad\qquad\qquad\qquad\qquad\qquad\qquad \\[0.05in]
+\sum_{n=1}^{4N_1+2N_2}\!\!\left(\!C_n(z)\!\left[\mu_{-2}'(\eta_n)+\left(\!D_n+\frac{1}{z-\eta_n}\right)\mu_{-2}(\eta_n)\right],
\,\widehat C_n(z)\left[\mu_{-1}'(\widehat\eta_n)+\left(\widehat D_n+\frac{1}{z-\widehat\eta_n}\right)\mu_{-1}(\widehat\eta_n)\right]\right).
\end{gathered}
\end{gather}
where $\int_\Sigma$ denotes the integral along the oriented contour shown in Fig. \ref{z-NZBC},
\begin{gather*}
C_n(z)=\frac{A[\eta_n]}{z-\eta_n}\,\mathrm{e}^{-2i\theta(\eta_n)}, \,\, D_n=B[\eta_n]-2\,i\,\theta'(\eta_n), \,\, \widehat C_n(z)=\frac{A[\widehat\eta_n]}{z-\widehat\eta_n}\,\mathrm{e}^{2i\theta(\widehat\eta_n)}, \,\, \widehat D_n=B[\widehat\eta_n]+2\,i\,\theta'(\widehat\eta_n),
\end{gather*}
$\mu_{-2}(x, t; \eta_n)$ and $\mu_{-2}'(x, t; \eta_n)$ are determined by $\mu_{-1}(\widehat\eta_n)$ and $\mu_{-1}'(\widehat\eta_n)$ as
\begin{align*}
\mu_{-2}(\eta_n)=\frac{iq_-}{\eta_n}\,\mu_{-1}\left(\widehat\eta_n\right), \quad
\mu_{-2}'(\eta_n)=-\frac{iq_-}{\eta_n^2}\,\mu_{-1}(\widehat\eta_n)+\frac{iq_-q_0^2}{\eta_n^3}\,\mu_{-1}'(\widehat\eta_n),
\end{align*}
and $\mu_{-1}(\widehat\eta_n)$ and $\mu_{-1}'(\widehat\eta_n)$ satisfy the following linear system of $8N_1+4N_2$ equations, $k=1, 2, \cdots, 4N_1+2N_2$,
\begin{gather}
\begin{gathered}
\begin{gathered}
\sum_{n=1}^{4N_1+2N_2}\widehat C_n(\eta_k)\,\mu_{-1}'(\widehat\eta_n)+\left[\widehat C_n(\eta_k)\left(\widehat D_n+\frac{1}{\eta_k-\widehat\eta_n}\right)-\frac{iq_-}{\eta_k}\,\delta_{k, n}\right]\mu_{-1}(\widehat\eta_n) \qquad\qquad \\[0.05in]
=-\mathrm{e}^{iv_-(x, t)\sigma_3}\begin{bmatrix}
\frac{iq_-}{\eta_k}\\[0.05in]1
\end{bmatrix}
-\frac{1}{2\pi i}\int_\Sigma\frac{\left(M^+J\right)_2(\zeta)}{\zeta-\eta_k}\,\mathrm{d}\zeta,
\end{gathered}  \\[0.1in]
\begin{gathered}
\sum_{n=1}^{4N_1+2N_2}\left(\frac{\widehat C_n(\eta_k)}{\eta_k-\widehat\eta_n}+\frac{iq_-q_0^2}{\eta_k^3}\,\delta_{k, n}\right)\mu_{-1}'(\widehat\eta_n)+\left[\frac{\widehat C_n(\eta_k)}{\eta_k-\widehat\eta_n}\left(\widehat D_n+\frac{2}{\eta_k-\widehat\eta_n}\right)-\frac{iq_-}{\eta_k^2}\,\delta_{k, n}\right]\mu_{-1}(\widehat\eta_n)\\[0.05in]
=-\mathrm{e}^{iv_-(x, t)\sigma_3}\begin{bmatrix}
\frac{iq_-}{\eta_k^2}\\[0.05in]0
\end{bmatrix}
+\frac{1}{2\pi i}\int_\Sigma\frac{\left(M^+J\right)_2(\zeta)}{\left(\zeta-\eta_k\right)^2}\,\mathrm{d}\zeta.
\end{gathered}
\end{gathered}
\end{gather}
\end{proposition}

\subsubsection{Reconstruction formula for the potential}

From the solution of the Riemann-Hilbert problem, one has the following asymptotic behavior of $M(x, t; z)$:
\begin{gather}
M(x, t; z)=\mathrm{e}^{iv_-(x, t)\sigma_3}+\frac{1}{z}\,M^{(1)}(x, t)+O\left(\frac{1}{z^2}\right), \quad z\to\infty,
\end{gather}
where
\begin{align}
\begin{aligned}
M^{(1)}(x, t)=&i\,\mathrm{e}^{iv_-(x, t)\sigma_3}\sigma_3\,Q_--\frac{1}{2\pi i}\int_\Sigma \left(M^+J\right)(x,t;\zeta)\,\mathrm{d}\zeta   \\[0.05in]
&+\sum_{n=1}^{4N_1+2N_2}\left[A[\eta_n]\,\mathrm{e}^{-2i\theta(\eta_n)}\left(\mu_{-2}'(\eta_n)+D_n\mu_{-2}(\eta_n)\right), A[\widehat\eta_n]\,\mathrm{e}^{2i\theta(\widehat\eta_n)}\left(\mu_{-1}'(\widehat\eta_n)+\widehat D_n\mu_{-1}(\widehat\eta_n)\right)\right].
\end{aligned}
\end{align}
Substituting $M(x, t; z)\mathrm{e}^{-i\theta(x, t; z)\sigma_3}$ into Eq.~(\ref{lax-x}) and matching $O\left(z\right)$ term, one obtains the reconstruction formula for the potential with double poles in the following proposition.
\begin{proposition}
The reconstruction formula for the double-pole solution of the DNLS equation (\ref{DNLS}) with NZBCs is
\begin{gather}
q(x, t)=\mathrm{e}^{iv_-(x, t)}\left(q_-\,\mathrm{e}^{iv_-(x, t)}-i\,\alpha^T\gamma\right),
\end{gather}
where the column vectors $\alpha=(\alpha^{(1)},\,\alpha^{(2)})^T$ and $\gamma=(\gamma^{(1)},\,\gamma^{(2)})^T$ are given by
\begin{gather*}
\begin{gathered}
  \alpha^{(1)}=\left(A[\widehat\eta_n]\,\mathrm{e}^{2i\theta(\widehat\eta_n)}\widehat D_n\right)_{\left(4N_1+2N_2\right)\times 1},\quad \alpha^{(2)}=\left(A[\widehat\eta_n]\,\mathrm{e}^{2i\theta(\widehat\eta_n)}\right)_{\left(4N_1+2N_2\right)\times 1},\\[0.05in]
  \gamma^{(1)}=\left(\mu_{-11}(\eta_n^*)\right)_{\left(4N_1+2N_2\right)\times 1},\quad \gamma^{(2)}=\left(\mu_{-11}'(\widehat\eta_n)\right)_{\left(4N_1+2N_2\right)\times 1}.
\end{gathered}
\end{gather*}
\end{proposition}

\subsubsection{Trace formulae and theta condition}

The trace formulae are
\begin{gather*}
\begin{gathered}
s_{11}(z)=\exp\left(-\frac{1}{2\pi i}\int_\Sigma\frac{\log\left[1-\rho(\zeta)\,\tilde\rho(\zeta)\right]}{\zeta-z}\,\mathrm{d}\zeta\right)\prod_{n=1}^{4N_1+2N_2}
\left(\frac{z-\eta_n}{z-\widehat\eta_n}\right)^2\mathrm{e}^{-iv_0},\\[0.05in]
s_{22}(z)=\exp\left(\frac{1}{2\pi i}\int_\Sigma\frac{\log\left[1-\rho(\zeta)\,\tilde\rho(\zeta)\right]}{\zeta-z}\,\mathrm{d}\zeta\right)\prod_{n=1}^{4N_1+2N_2}
\left(\frac{z-\widehat\eta_n}{z-\eta_n}\right)^2\mathrm{e}^{iv_0}.
\end{gathered}
\end{gather*}
Let $z\to 0$, the theta condition is obtained. That is to say, there exists $j\in\mathbb{Z}$ such that
\begin{gather}
\mathrm{arg}\left(\frac{q_-}{q_+}\right)+2v=16\sum_{n=1}^{N_1}\mathrm{arg}(z_n)+8\sum_{m=1}^{N_2}\mathrm{arg}(w_m)+2j\pi+\frac{1}{2\pi}\int_\Sigma\frac{\log\left(1-\rho(\zeta)\tilde\rho(\zeta)\right)}{\zeta}\,\mathrm{d}\zeta.
\end{gather}

\subsubsection{Reflectionless potential: double-pole soliton solutions}

We consider the reflectionless potential. From the Jost integrable equation, one derives $\varPhi_{\pm}(x, t; q_0)=E_{\pm}(q_0)$. Combining with the definition of scattering matrix, one has $S(q_0)=I$ and $q_+=q_-$. From the theta condition, there exist $j\in\mathbb{Z}$ such that
\begin{gather}
v_0=8\sum_{n=1}^{N_1}\mathrm{arg}(z_n)+4\sum_{m=1}^{N_2}\mathrm{arg}(w_m)+j\pi.
\end{gather}

From the reconstruction formula, one can deduce the reflectionless potential in terms of determinants:
\begin{gather}\label{wufanshe1-1}
q(x, t)=\left(1+\frac{\mathrm{det}\left(G\right)}
{\mathrm{det}\left(H\right)}\right)q_-\,\mathrm{e}^{2iv_-(x, t)}, \quad G=\begin{bmatrix}
H&\beta\\[0.02in] \alpha^T&0\end{bmatrix},
\end{gather}
where $\left(8N_1+4N_2\right)\times\left(8N_1+4N_2\right)$ matrix $H$ is defined as
\begin{gather*}
H=\begin{bmatrix}
H^{(1, 1)}& H^{(1, 2)}\\[0.02in]
H^{(2, 1)}& H^{(2, 2)}
\end{bmatrix},\quad H^{(m, j)}=\left(h^{(m, j)}_{i, n}\right)_{\left(4N_1+2N_2\right)\times\left(4N_1+2N_2\right)},\quad m, j=1, 2  \\[0.05in]
h^{(1, 1)}_{i, n}=\widehat C_n(\eta_i)\left(\widehat D_n+\frac{1}{\eta_i-\widehat\eta_n}\right)-\frac{iq_-}{\eta_i}\,\delta_{i, n}, \quad h^{(1, 2)}_{i, n}=\widehat C_n(\eta_i), \\
h^{(2, 1)}_{i, n}=\frac{\widehat C_n(\eta_i)}{\eta_i-\widehat\eta_n}\left(\widehat D_n+\frac{2}{\eta_i-\widehat\eta_n}\right)-\frac{iq_-}{\eta_i^2}\,\delta_{i, n}, \quad h^{(2, 2)}_{i, n}=\frac{\widehat C_n(\eta_i)}{\eta_i-\widehat\eta_n}+\frac{iq_-q_0^2}{\eta_i^3}\,\delta_{i, n},
\end{gather*}
and $8N_1+4N_2$ column vector $\beta$ is
\begin{gather*}
\beta=\begin{bmatrix}
\beta^{(1)}\\\beta^{(2)}
\end{bmatrix},\quad
\beta^{(1)}=\left(\frac{1}{\eta_i}\right)_{\left(4N_1+2N_2\right)\times 1}, \quad \beta^{(2)}=\left(\frac{1}{\eta_i^2}\right)_{\left(4N_1+2N_2\right)\times 1}.
\end{gather*}
This formula is an implicit solution since it contains the term $v_-(x,t)$. To yield the explicit solution, from Jost integrable equation and trace formula, one has
\begin{gather}
\begin{aligned}
M(&x, t; z)=E_-(z)\\
&+k(z)\sum_{n=1}^{4N_1+2N_2}\left(\frac{\mathop\mathrm{P_{-2}}\limits_{z=\eta_n}\left[M(z)/k(z)\right]}{\left(z-\eta_n\right)^2}+\frac{\mathop\mathrm{Res}\limits_{z=\eta_n}\left[M(z)/k(z)\right]}{z-\eta_n}+\frac{\mathop\mathrm{P_{-2}}\limits_{z=\widehat\eta_n}\left[M(z)/k(z)\right]}{\left(z-\widehat\eta_n\right)^2}+\frac{\mathop\mathrm{Res}\limits_{z=\widehat\eta_n}\left[M(z)/k(z)\right]}{z-\widehat\eta_n}\right),
\end{aligned}
\end{gather}
with which one can solve the $\gamma$ explicitly. Sustituting $\gamma$ into the reconstruction formula for the reflectionless potential, one can pose another formula as
\begin{gather}\label{wufanshe1-2}
q(x, t)=\left(\mathrm{e}^{iv_-(x, t)}+\frac{\mathrm{det}\left(\widehat G\right)}
{\mathrm{det}\left(\widehat H\right)}\right)q_-\,\mathrm{e}^{iv_-(x, t)}, \quad \widehat G=\begin{bmatrix}
\widehat H&\beta\\[0.02in] \alpha^T&0\end{bmatrix},
\end{gather}
where $\left(8N_1+4N_2\right)\times\left(8N_1+4N_2\right)$ matrix $\widehat H$ is defined as
\begin{gather*}
\widehat H=\begin{bmatrix}
\widehat H^{(1, 1)}& \widehat H^{(1, 2)}\\[0.02in]
\widehat H^{(2, 1)}& \widehat H^{(2, 2)}
\end{bmatrix},\quad \widehat H^{(m, j)}=\left(\widehat h^{(m, j)}_{i, n}\right)_{\left(4N_1+2N_2\right)\times\left(4N_1+2N_2\right)},\quad m, j=1, 2  \\[0.05in]
\widehat h^{(1, 1)}_{i,n}=\frac{k(\eta_i)}{k(\widehat \eta_n)}\widehat C_n(\eta_i)\left(\widehat D_n+\frac{1}{\eta_i-\widehat\eta_n}-\frac{k'(\widehat\eta_n)}{k(\widehat\eta_n)}\right)-\frac{iq_-}{\eta_i}\,\delta_{i, n}, \quad \widehat h^{(1, 2)}_{i, n}=\frac{k(\eta_i)}{k(\widehat \eta_n)}\widehat C_n(\eta_i), \\
\widehat h^{(2,1)}_{i, n}=\frac{\widehat C_n(\eta_i)}{k(\widehat\eta_n)}\left[\frac{k(\eta_i)}{\eta_i-\widehat\eta_n}\left(\widehat D_n+\frac{2}{\eta_i-\widehat\eta_n}-\frac{k'(\eta_i)}{k(\widehat\eta_n)}\right)-k'(\eta_i)\left(\widehat D_n+\frac{1}{\eta_i-\widehat\eta_n}-\frac{k'(\widehat\eta_n)}{k(\widehat\eta_n)}\right)\right]-\frac{iq_-}{\eta_i^2}\,\delta_{i, n}, \\
\widehat h^{(2, 2)}_{i, n}=\frac{\widehat C_n(\eta_i)}{k(\widehat\eta_n)}\left(\frac{k(\eta_i)}{\eta_i-\widehat\eta_n}-k'(\eta_i)\right)+\frac{iq_-q_0^2}{\eta_i^3}\,\delta_{i, n},
\end{gather*}

From Eqs.~(\ref{wufanshe1-1}) and (\ref{wufanshe1-2}), we have

\begin{theorem} The explicit double-pole soliton solutions of the DNLS equation (\ref{DNLS}) with NZBCs are found as
\begin{gather}
q(x, t)=\left(\frac{\mathrm{det}\left(\widehat G\right)}{\mathrm{det}\left(\widehat H\right)}\right)^2\frac{\mathrm{det}\left(H\right)}{\mathrm{det}\left(G\right)}\left(1+\frac{\mathrm{det}\left(H\right)}{\mathrm{det}\left(G\right)}\right)q_-.
\end{gather}
\end{theorem}

\begin{figure}[!t]
\centering
\includegraphics[scale=0.5]{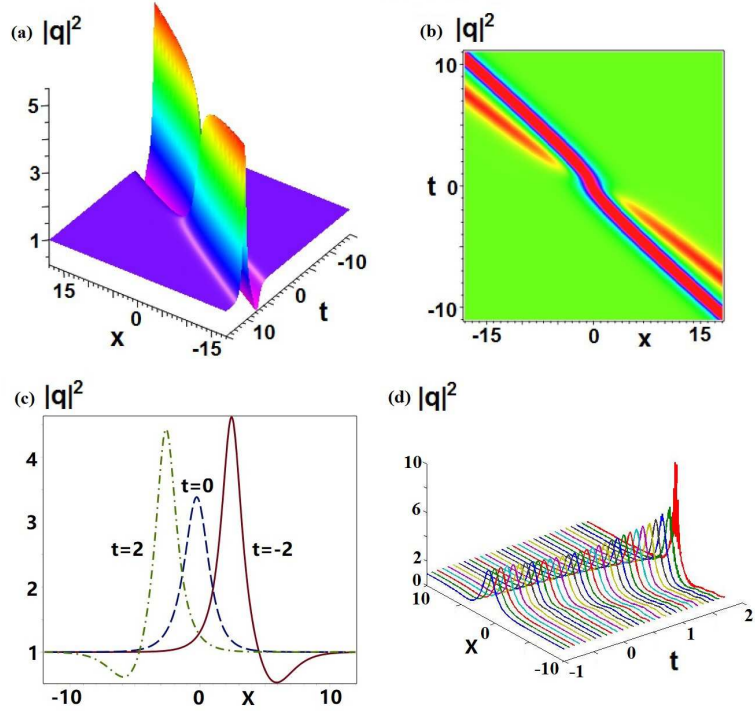}
\caption{Double-pole solutions with NZBCs. (a, b) $N_1=0, N_2=1, q_{\pm}=1, w_1=\exp\left(\frac{\pi}{4}i\right), A[w_1]=i, B[w_1]=1+\left(1-\sqrt{2}\right)i$;  (c, d) $N_1=1, N_2=0, q_{\pm}=1, z_1=2\exp\left(\frac{\pi}{6}i\right), A[w_1]=i, B[w_1]=i.$}
\label{fig4}
\end{figure}

For example, we have the following double-pole solutions:

\begin{itemize}

 \item  When $N_1=0, N_2=1, q_{\pm}=1, w_1=\exp\left(\frac{\pi}{4}i\right), A[w_1]=i, B[w_1]=1+\left(1-\sqrt{2}\right)i$, we have the double-pole dark-bright soliton, $q(x,t)=P_{11}/P_{12}$ with
\bee
\begin{array}{rl}
P_{11}=&\!\!\!\Big\{2[1+i-\sqrt{2}(1+(1+i)t)]e^{4t+2x}+2i[t^2+t+1+\sqrt{2}(t+1/2)]e^{2t+x} \v\\
  &\!\!\!+2[1-i+\sqrt{2}(i+(i-1)t)]\Big\}^2\Big\{[\sqrt{2}(1-2i+2t)+2(i-t^2)+2(2i-1)t]e^{4t+2x}\v\\
  &\!\!\!+(i\sqrt{2}t-i+\sqrt{2})e^{6t+3x}+(e^{8t+4x}+1)/2+(\sqrt{2}(it+1+i)-i)e^{2t+x}\Big\} \v\\
P_{12}=&\!\!\!\Big\{[\sqrt{2}((1+i)t+1)-1-i]e^{4t+2x}+[2\sqrt{2}(1-i+2t)-2+2i-4t^2+4(i-1)t]e^{2t+x} \v\\
 &\!\!\!+\sqrt{2}(i+(i-1)t)+1-i\Big\}^2\Big\{[t^2+t+1-\sqrt{2}(t+1/2)]e^{4t+2x}+(1-\sqrt{2}t)e^{6t+3x}/2\v\\
 &\!\!\!+(e^{8t+4x}+1)/4+[\sqrt{2}(t+1)-1]e^{2t+x}/2\Big\}.
 \end{array}
 \ene
which is a semi-rational soliton, and differs from the single-pole solutions usually expressed by the exponential functions even if the double-pole soliton displays the interaction of dark and bright solitons (see Figs.~\ref{fig4}(a)-(e)). Fig.~\ref{fig4}c displays that when $t=0$ the Gaussian-ilke profile with non-zero boundary, when $t\not=0$,

We use the exact double-pole bright-dark soliton at $t=-1$ as the initial condition without a small noise to numerically check the wave propagation such that we find that the wave stably propagates in a short time, and then the amplitude begins to become larger and larger (see Fig.~\ref{fig4}e).

 \item As $N_1=1, N_2=0, q_{\pm}=1, z_1=2\exp\left(\frac{\pi}{6}i\right), A[w_1]=i, B[w_1]=i$, we have the interaction of two breathers, which is complicated and omitted here (see Fig.~\ref{fig5}).

\end{itemize}

\begin{figure}[!t]
\centering
\vspace{0.1in}\includegraphics[scale=0.65]{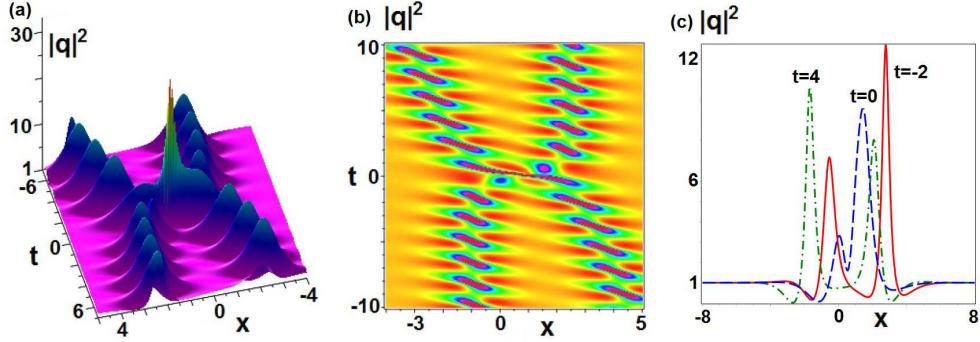}
\caption{Double-pole solutions with NZBCs. (a, b) $N_1=0, N_2=1, q_{\pm}=1, w_1=\exp\left(\frac{\pi}{4}i\right), A[w_1]=i, B[w_1]=1+\left(1-\sqrt{2}\right)i$;  (c, d) $N_1=1, N_2=0, q_{\pm}=1, z_1=2\exp\left(\frac{\pi}{6}i\right), A[w_1]=i, B[w_1]=i.$}
\label{fig5}
\end{figure}

{\it Remark.} The obtained $N$ double-pole solitons of Eq.~(\ref{DNLS}) can be applied to the modified NLS Eq.~(\ref{dnlsg}) by the gauge transformation.

\section{Conclusions and discussions}

In conclusions, we have presented the inverse scattering transforms for the DNLS equation with double poles under ZBC and NZBCs at infinity.
A rigorous theory for direct and inverse problems was proposed. The direct scattering illustrates the analyticity, symmetries, discrete and  spectrum and asymptotic behavior. The inverse problem is formulated and solved via a Riemann-Hilbert problem, which derives the trace formula and reflectionless potential. In addition, the reflectionless potential with double poles is deduced explicitly by determinants.  Some representative semi-rational bright-bright soliton, dark-bright soliton, and breather-breather solutions  are
examined in detail. Moreover, we will study long-time asymptotics for the DNLS equation with NZBCs via the modified Deift-Zhou method~\cite{MDZ} in another literature.

\vspace{0.1in}
\baselineskip=15pt

\noindent {\bf Acknowledgements}

\vspace{0.05in}
This work was partially supported by the NSFC under grants Nos.11731014 and 11571346, and CAS Interdisciplinary Innovation Team.

\end{document}